\pdfoutput=1
\documentclass[pdftex,pra,aps,twoside,nofootinbib,showkeys,showpacs,twocolumn]{revtex4}
\usepackage{setspace}
\usepackage{float}
\usepackage{indentfirst}
\usepackage{url}
\usepackage[stable]{footmisc}
\usepackage{fancyhdr}
\usepackage{appendix}
\usepackage{dsfont}
\usepackage{epstopdf}
\usepackage[all]{xy}
\usepackage{ifpdf}
\usepackage{color}
\ifpdf
 \usepackage[pdftex]{graphicx}
 \else
\usepackage{graphicx}
 \fi
\usepackage{graphicx}
\usepackage[T1]{fontenc}
\usepackage{textcomp}
\usepackage{newcent}
\usepackage[sc,noBBpl]{mathpazo}
\usepackage{amsmath,amsfonts,amssymb,amsthm, graphics}
\usepackage[mathscr]{eucal}
\usepackage{appendix}

\theoremstyle{plain}
\newtheorem{theorem}{Theorem}
\newtheorem{lemma}[theorem]{Lemma}
\newtheorem{proposition}[theorem]{Proposition}
\newtheorem{corollary}[theorem]{Corollary}
\newtheorem{conjecture}[theorem]{Conjecture}

\newtheorem{fact}[theorem]{Fact}

\newtheoremstyle{note}{\topsep}{\topsep}{\slshape}{}{\scshape}{}{ }{}
\theoremstyle{note}

\newtheorem{remark}[theorem]{Remark}
\newcommand\tr{\operatorname{Tr}}

\newcommand\id{\operatorname{\mathrm{Id}}}

\newcommand{\<}{\langle}
\renewcommand{\>}{\rangle}

\newcommand{\e}{\mathrm{e}}
\newcommand\hcal{\cH}
\newcommand\be{\begin{equation}}
\newcommand\ee{\end{equation}}
\newcommand\bea{\begin{array}}
\newcommand\eea{\end{array}}
\newcommand\ben{\begin{eqnarray}}
\newcommand\een{\end{eqnarray}}
\newcommand\ot{\otimes}

\newcommand\bei{\begin{itemize}}
\newcommand\eei{\end{itemize}}
\newcommand\bee{\begin{enumerate}}
\newcommand\eee{\end{enumerate}}

\newcommand{\E}{\operatorname{e}}

\newcommand{\mathsym}[1]{{}}
\newcommand{\unicode}[1]{{}}

\def\<{\langle}
\def\>{\rangle}
\def\ot{\otimes}
\def\ermax{E_R^{\max}}
\def\esmax{E_S^{\max}}
\def\ecal{{\cal E}}
\def\hcal{{\cal H}}
\def\ecalr{{\ecal_R}}
\def\etaE{\eta_{E-E_S}^R}

\def\ideta{\eta}
\def\rhor{\rho_R}
\def\rhos{\rho_S}

\def\poscrit{DMP}
\def\dampmat{Damping Matrix}
\def\gto{Enhanced Thermal Operations}

\newsavebox{\smlmat}
\savebox{\smlmat}{$\left[\begin{smallmatrix}p&\alpha \\\alpha^{*}&\widetilde{p}\end{smallmatrix}\right]$}

\begin{document}
\title{Towards fully quantum second laws of thermodynamics: limitations on the evolution of quantum coherences}

\author{Piotr \'Cwikli\'nski$^{1,2}$, Micha{\l} Studzi\'nski$^{1,2}$,
Micha{\l} Horodecki$^{1,2}$ and Jonathan Oppenheim$^{3}$}
\affiliation{
$^1$ Institute of Theoretical Physics and Astrophysics, University of Gda\'nsk, 80-952 Gda\'nsk, Poland \\
$^2$ National Quantum Information Centre of Gda\'nsk, 81-824 Sopot, Poland \\
$^3$ Department of Physics and Astronomy, University College of London, and London Interdisciplinary Network for Quantum Science, London WC1E 6BT, UK
}

\date{\today}

\begin{abstract}
The second law of thermodynamics places a limitation into which states a system can evolve into. For systems in contact with a heat bath, it can be combined with the law of energy conservation, and it says that a system can only evolve into another if the free energy goes down. Recently, it's been shown that there are actually many second laws, and that it is only for large macroscopic systems that they all become equivalent to the ordinary one. These additional second laws also hold for quantum systems, and are, in fact, often more relevant in this regime. They place a restriction on how the probabilities of energy levels can evolve. Here, we consider additional restrictions on how the coherences between energy levels can evolve. Coherences can only go down, and we provide a set of restrictions which limit the extent to which they can be maintained. We find that coherences over energy levels must decay at rates that are suitably adapted to the transition rates between energy levels. We show that the limitations are matched in the case of a single qubit, in which case we obtain the full characterization of state-to-state transformations. For higher dimensions, we conjecture that more severe constraints exist. We also introduce a new class of thermodynamical operations which allow for greater manipulation of coherences and study its power with respect to a class of operations known as thermal operations.
\end{abstract}

\maketitle

We consider a quantum system in state $\rho_S$ which can be put in contact with a reservoir at temperature $T$. The second law of thermodynamics, combined with the first law (conservation of energy) states that the free energy
\begin{align}
F=\tr(H\rho_S)-TS(\rho_S)
\label{eq:freeenergy}
\end{align}
can only decrease, where $S(\rho_S)$ is the von Neumann entropy\footnote{One can also take the course grained entropy, but since we are interested in small quantum systems where all degrees of freedom can be precisely measured in principle, taking the von Neumann entropy will help us answer the questions we're interested in here.} $S(\rho)=-\tr\rho\log\rho$ and $H$ is the system's Hamiltonian. Although this is a necessary constraint on what state transformation are possible, we now know that it is not sufficient. For transitions between two states, diagonal in the energy basis, there are a set of necessary and sufficient conditions which must be satisfied in order for a state to transform into another state. One has a family of free energies in the case of catalytic processes\cite{2013arXiv1305.5278B} (i.e. where one is allowed an ancilla
which can be returned to its original state in the spirit of Clausius-Planck formulation of the Second Law). For non-catalytic transformations, the set of necessary and sufficient conditions were proven to be majorization\cite{horodecki_reversible_2003} in the case when the Hamiltonian is $H=0$ and thermo-majorization\cite{Horodecki_2013thermo} in general. It is only in the thermodynamic limit that all these conditions become equivalent to the ordinary second law
of equation \eqref{eq:freeenergy}. However, for single finite systems (sometimes called the single-shot scenario), the full set of conditions are relevant. It is this finite regime which is more relevant for quantum systems or
even in the meso-scopic regime, especially if long range interactions are present.

Regarding states which are not diagonal in the energy basis, thermo-majorization (or the generalised free energies of \cite{2013arXiv1305.5278B} in the catalytic case) are still necessary conditions for state transformations and
place conditions on the diagonal entries of the state $\rho_S$ (where we assume that $\rho_S$ is written in the energy eigenbasis).
But these conditions do not say anything about how off-diagonal elements between different energies behave. 
Partial results were obtained in \cite{Horodecki_2013thermo,2013arXiv1305.5278B} for the case where only the input state is non-diagonal
and simultaneously posted with this paper, in \cite{2015NatCo...6E6383L}, where relations between purity and quantum asymmetry in the spirit of coherences have been formulated and the authors obtained the "free-energy relation" for coherences (second law). However finding a complete set of quantum limitations is still a challenge.
Here, by providing a first systematic approach to coherences, we will present a set of conditions, called Damping Matrix Positivity (\poscrit). Unlike the results of \cite{Horodecki_2013thermo,2013arXiv1305.5278B}, they are not necessary and sufficient, although we will show that they are for the case of a qubit.

Since we are interested in studying fundamental limitations, we allow for the experimenter to perform the largest possible set of operations within the context of thermodynamics. Namely, we allow them to have access to a heat 
bath at temperature $T$, and to perform any unitary. Since the laws of physics must conserve energy, and we are interested in how energy flows in thermodynamics, the unitary should conserve the total energy of all systems it acts on
but this is the only restriction.  This provides a precise definition of what we mean by thermodynamics, recasting it as a resource theory known as Thermal Operations (TO). This was introduced in \cite{janzing_thermodynamic_2000} (cf. \cite{Streater_dynamics}) and applied later in \cite{Horodecki_2013thermo, 2011arXiv1111.3882B} where the addition of a work system enabled one to compute the work required for a state transformation.
We will define these operations more carefully and then derive the restrictions they impose on state transformations. In particular, we will present our conditions, and
discuss it in details for qubits, where we see that \poscrit\ is a necessary and sufficient condition for state transformations.
As a result, we fully characterize the qubit-qubit case, as well as provide
limitations for higher dimensional states.
Then, we introduce a new class of operations we call \gto, and study its power with respect to TO. They appear to be more powerful, in that for them, \poscrit\ are necessary and sufficient conditions for state transformations, while for TO we believe \poscrit\ are
not sufficient.  At least in the qubit case, TO can be describe by three conditions: completely positive trace preserving, some commutation relation, and preservation of the Gibbs state. We obtain a part of our findings by adapting results for studies of the weak-coupling between the system and the heat-bath, and dynamical semi-groups \cite{Roga2010, Alicki_book, Alicki_76, Streater_dynamics}.

{\it Thermodynamics as a resource theory and Thermal Operations}
In order to derive any laws of thermodynamics we need to say what thermodynamics is -- in other words, define the class of operations which constitute thermodynamics.
Thermodynamics is then viewed as a resource theory \cite{janzing_thermodynamic_2000, horodecki_reversible_2003, Horodecki_2013thermo,2011arXiv1111.3882B, gour2013resource}.
In the resource theory, one considers some class of operations, and then asks how much of some resource can be used to perform the desired task and how this resource can be manipulated. In the case of thermodynamics, it is viewed as a theory involving state transformations in the presence of a thermal bath. To describe it, one can then exploit some mathematical machinery from single-shot information theory, where one does not have access to many copies of independent and identically distributed bits of information\cite{RR-phd}. 

\begin{figure}
  \includegraphics[width=0.45\textwidth]{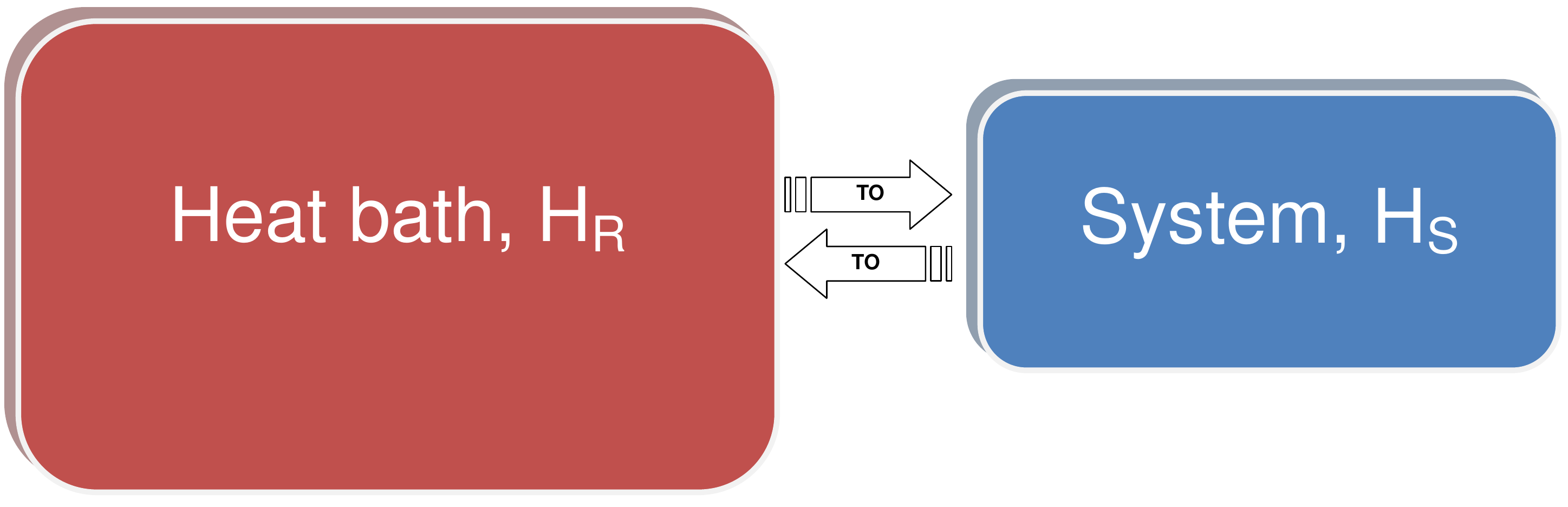}
  \caption{One considers a system $S$ in a quantum state $\rho_S$ with a fixed Hamiltonian $H_S$, in contact with the thermal reservoir (heat bath) $R$ in a Gibbs state $\tau_{R}$ (possibly many copies of it) with Hamiltonian $H_R$ - acting as a free resource. Interactions (white arrows) between them, are implemented under the paradigm of Thermal Operations (TO), i.e., by energy preserving, unitary operations $U$, commuting with the total Hamiltonian. The goal is to obtain some other state $\sigma_{S}$. The energy spectrum of the bath is highly degenerated (small or no degeneracy drastically reduce the set of Thermal Operations) and its maximal energy will tend to infinity. We also make an assumption that the dimension of the bath is much larger than that of the system. Moreover, initially, the total system is in the product state of the bath and system $\tau_R \ot \rho_S$.}
\label{fig:thermalop}
\end{figure}

We wish to explore fundamental limitations on
state transformations, therefore we should allow the experimenter to perform any unitary transformation.  
However, any interaction allowed in nature has to conserve energy if we consider the total system, and since
thermodynamics requires precise accounting of all sources of energy, the unitary must conserve energy. Of course,
we often consider adding interaction Hamiltonians, or performing unitaries which do not conserve energy, but this is
only because we are ignoring degrees of freedom which, if there change in energy was taken into account, would restore energy conservation. Here, we need to include these additional systems, not only because we want to account for
all sources of energy, but because we want to understand coherences and these additional systems may contain coherences which could be transferred to our system. We thus consider all systems with coherence as being part of the 
system. Indeed all the standard thermodynamical paradigms we are interested in, can be made to fall within thermodynamics in this manner~\cite{2011arXiv1111.3882B}. We can thus use Thermal Operations (TO)to study fundamental limitations on the manipulation of coherences.  
The TO paradigm preassumes that there is a heat bath, described by a Gibbs state and helps to describe what can happen with a system which can interact with the heat bath. It also treats the microscopic system, without any approximations. 

Formally, under TO one can
\begin{enumerate}
\item Bring in an arbitrary system in a Gibbs state with temperature T (free resource).
\item Remove (discard) any system.
\item Apply a unitary that commutes with the total Hamiltonian.
\end{enumerate}
The class of TO is generated by the unitaries $U$ (which act on the system, bath and other ancillas) which obey the energy conservation condition
\begin{equation}
	[U, H_S + H_R + H_W] = 0,
\end{equation}
where $H_W$ is a work system or a clock, or any other object under consideration besides the system and bath. Equation \eqref{super} is necessary and sufficient if we wish to ensure that energy is conserved on every input state\cite{Horodecki_2013thermo}.  This is natural if we wish to apply our thermal machine on arbitrary unknown quantum states. Thus, an arbitrary Thermal Operation is obtain by the implementation of an energy-preserving $U$ and tracing out the heat bath (see also, Fig. \ref{fig:thermalop}). Precisely, $\Lambda \in$ TO, when
\be \Lambda(\rho_S) = \tr_R (U \tau_R \ot \rho_S U^{\dagger}). \label{lambda} \ee

It is worth noting that it is Eq. \eqref{super} that prevents one from creating coherences over energy levels if one doesn't already start with them. One can extend TO to the case where one is allowed as a resource, a reference frame which acts as a source of infinite coherence, and in such a case, one can lift the superselection imposed by Eq. \eqref{super} (in the context of thermodynamics, see \cite{2011arXiv1111.3882B,Aberg_2013}).

{\it Allowed state-to-state transformations under Thermal Operations}
As we already mentioned TO cannot create coherent superpositions between
eigenstates, but what are the ultimate limitations for a general $(\rho, H_S) \rightarrow (\sigma, H_S)$
state-to-state transition? In \cite{Horodecki_2013thermo} necessary and sufficient
conditions, in terms of monotones, have been put forward for the block diagonal entries of a state written in the energy basis. These conditions are discussed in Supplemenatry Note. Here, by noticing some general properties of TO, we will provide bounds for off-diagonal elements - {\it coherences} under the assumption that the system Hamiltonian $H_S$ has non-degenerate Bohr spectrum, i.e., there are no degeneracies in the nonzero differences of energy levels of the Hamiltonian. To obtain some of our results, we will adapt results derived for open systems, precisely, for Davies maps under weak-coupling \cite{Roga2010, Alicki_book, Alicki_76}.

{\it Properties of Thermal Operations}
Let us examine properties of TO. 
First, the diagonal elements of a density matrix are not mixed with off-diagonal ones during evolution under TO (i.e. they evolve independently). Moreover, for systems having non-degenerated Bohr spectra, coherences are not mixed among themselves. We can thus say that TO are {\it block-diagonal}, i.e, for an off-diagonal (diagonal) element $|i\>_S\<j|$ ($|i\>_S\<i|$) of state $\rho_S$ one gets (for proofs, see, Supplementary Note 2)
\be \Lambda (|i\>_S\<j|) = \alpha_{ij} |i\>_S\<j|, i \neq j \label{main1} \ee
and
\be \Lambda (|i\>_S\<i|) = \sum_{ij} p(i \rightarrow j) |j\>_S\<j| \label{main2}, \ee
where $\Lambda$ is defined as in Eq. \eqref{lambda}, $\alpha_{ij}$ are factors by which the off-diagonal elements are multiplied (damped) during the transition, and $p(i \rightarrow j)$ is a transition
probability of moving element $i$'s into $j$'s and $p(i)$ is a probability of occupying an energy state $i$.

TO are physical operations, so the dynamics should be implemented by completely positive trace preserving maps (CPTP maps). Together with the fact that under TO, the Gibbs state is preserved, we have a set of properties fulfilled by TO. It is known \cite{Roga2010, Alicki_book} that these properties are also satisfied by Davies maps appearing in the weak-coupling regime for Hamiltonians with non-degenerate Bohr spectra. Using the above properties, we obtain constraints for the behavior of coherences. We thus get bounds for off-diagonal elements which are determined by the probability for staying in the same energy-level.

{\it Quantum states - second laws for off-diagonal elements}
We will now use the above properties of TO, to study allowed states transitions. From the property given by Equation  \eqref{main1}-\eqref{main2}, we obtain that there exist two families of bounds, one for diagonal elements of states (thermo-majorization) and the second one for coherences.

Suppose now, that somehow we can transform the diagonal of an input $d-$ level state into another $d-$level state with some other diagonal entries. Our main question is then: {\it How does this process affect coherences, i.e., the off-diagonal elements?}

To answer it, let us use other properties of TO. As shown in Ref \cite{Roga2010}, the property of CPTP combined with formulas \eqref{main1} and \eqref{main2} imply that the following matrix must be positive:
\be
\begin{bmatrix} \label{main}
p(0 \rightarrow 0) & \alpha_{01} & \ldots & \alpha_{0 n-1}\\
\alpha_{10} & p(1 \rightarrow 1) & \ldots & \alpha_{1 n-1}\\
\vdots & \vdots & \ddots & \vdots \\
\alpha_{n-1 0} & \alpha_{n-1 2} & \ldots & p(n-1 \rightarrow n-1)
\end{bmatrix} \geq 0.
\ee
We will call the above matrix the \dampmat, and the above condition, Damping Matrix Positivity or \poscrit.
Let us note, that the matrix is crucial for processing coherences. Indeed, positivity implies that the damping factors in particular satisfy
\be
\label{minors}
|\alpha_{ij}| \leq \sqrt{p(i \rightarrow i)p(j \rightarrow j)}.
\ee
Thus, the coherences must be damped at least by a factor $\sqrt{p(i \rightarrow i)p(j \rightarrow j)}$ that comes from the $2 \times 2$ minors of the matrix from Eq. \eqref{main}. Since the present paper appeared on the arXiv, the formulas has been generalized to the case of an arbitrary spectrum in \cite{Jennings_thermo}. In the subsequent section we will show that for qubits, this is the only constraint for processing coherences by TO.

{\it Qubit example}
For qubits, we have necessary and sufficient sets of criteria, by showing that for a
given process on a diagonal, the damping factor (for coherences) from Eq. \eqref{minors} can always be equal to $\sqrt{p(i \rightarrow i)p(j \rightarrow j)}$. We will determine this optimal factor for an arbitrary $\rho \rightarrow \sigma$ transition.

Going into details, consider two states $\rho_S = \begin{bmatrix}
p & \alpha\\
\alpha^{*} & 1-p
\end{bmatrix}$ and $\sigma_S = \begin{bmatrix}
q & \chi\\
\chi^{*} & 1-q
\end{bmatrix}$, written in the energy eigenbasis of the Hamiltonian system $H_S$, where $*$ stands for complex conjugation. We know that the evolution of diagonal elements can be separated from off-diagonal ones, so for the former one uses thermo-majorization (leading to four different situations discussed in Supplementary Note 3). For the latter we obtain that the decaying rate of coherences depends only on the diagonal transition rates (and consequently, on elements of states and energies associated with the Hamiltonian of the system). Namely, coherences obey the following inequality

\be |\chi| \leq  |\alpha| \kappa, \label{eq:main} \ee where
\be \kappa =  \frac{\sqrt{(q-\widetilde{p}\E^{\beta \Delta E})(p-\widetilde{q}\E^{\beta \Delta E})}}{\left| p-\widetilde{p}\E^{\beta \Delta E} \right|},
\ee
$\widetilde{q} = 1 - q$, $\widetilde{p} = 1 - p$, $\E^{\pm\beta \Delta E} = \E^{\pm\beta (E_j-E_i)}$ with $E_i$ being the energy of the system and $\beta = \frac{1}{kT}$ the inverse temperature. Note that the phases commute with the total Hamiltonian of our setup, so we can restrict our attention only to moduli of the coherences. We have necessary and sufficient conditions for arbitrary qubit $\rho \rightarrow \sigma$ transitions under TO, where for \textit{a)} diagonal elements we use thermo-majorization, \textit{b)} for coherences, Eq. \eqref{eq:main} (in Supplementary Note 4 we show that it can be achieved with equality). In the appendix, we also express our damping factor in terms of relaxation times ($T_2$) \cite{Gorini1978149, PhysRevA.66.062113}.

{\it Sufficiency of the second laws?}
It is clear that for an arbitrary transitions, there are many stochastic maps that lead to the same final state and each such map can be implemented by possibly many unitary transformations. We need such unitaries that damp as little as possible, the off-diagonal elements of the density matrix - for which, the inequalities coming from the $2\times 2$ minors of the Choi map from Eq. \eqref{minors} are all saturated. This would optimise the preservation of coherences. But, is it always possible? As we have shown, for qubits, for every state-to-state transition, we have only one channel that realizes it and we can always make the inequality that gives us a dumping factor for coherences tight. This uniqueness of channel may be not true anymore for higher dimensional states. In Supplementary Note 6, we choose a qutrit state-to-state transition $(0, \frac{1}{2}, \frac{1}{2})\rightarrow(\frac{\E^{-\beta \Delta E_{21}}}{2}, \frac{1- \E^{-\beta \Delta E_{21}} + \E^{-\beta \Delta E_{20}}}{2}, \frac{1 - \E^{-\beta \Delta E_{20}}}{2})$, which can only be realised by a unique set of transition probabilities. For this set of transition probabilities, one is not able to find a unitary map that at the same time realizes the exact states transition and leads to the saturation of bounds for coherence preservation. 

{\it Enhanced Thermal Operations: a class of operations which saturate the \poscrit criteria}
As we already have observed, in the case when one considers Hamiltonians with nondegenerated Bohr spectra, the properties of TO 
used in this paper are similar to those occurring when one studies Davies maps for many-level systems. We shall now introduce a class of operations that is defined by these properties. We will call this class \gto.

We define Enhanced Thermal Operations (ETO) in the following way. $\Lambda \in ETO$ when
\begin{enumerate}
\item $[\Lambda, \hat{H_S}] = 0$.
\item It is CPTP.
\item It preserves the Gibbs state.
\end{enumerate}
Here $\hat H_S$ is a superoperator defined   by $ \hat{H}_S(X) = [H_S, X]$ for all operators $X$. 
The first property gives us that under ETO, one is able to realize all possible transformations 
that satisfy the constraints given by  Eq. \eqref{main}. We will use our previous findings to compare the power of these two classes.

We show that for qubits, TO are equal to the new class and as a result we have laws for any state-to-state transition under TO. Essentially, for qubits, TO can already saturate the
bound given by Eq. \eqref{main} and thus do no worse than \gto. For qutrits, we provided a family of initial and final states $\rho$ and $\sigma$ that by \gto, one can transform $\rho$ into $\sigma$ exactly, but it is not possible under TO. Based on this, one can try to conclude that TO are outperformed by \gto, and what is more, state transitions by TO are not equal to the ones under \gto\ (the latter statement is stronger, it could be that the set of TO is smaller than \gto, but both classes lead to the same laws of transformations). However, it is not a conclusive result, because one can try to approximate the channel that is used to realize the transition under TO from the previous section, which may lead to the saturation of the bound for optimal coherences preservation. See also, Fig. \ref{fig:sets}.
\begin{figure}[h]
\begin{center}$
\begin{array}{cc}
\includegraphics[width=0.20\textwidth, height=0.20\textheight]{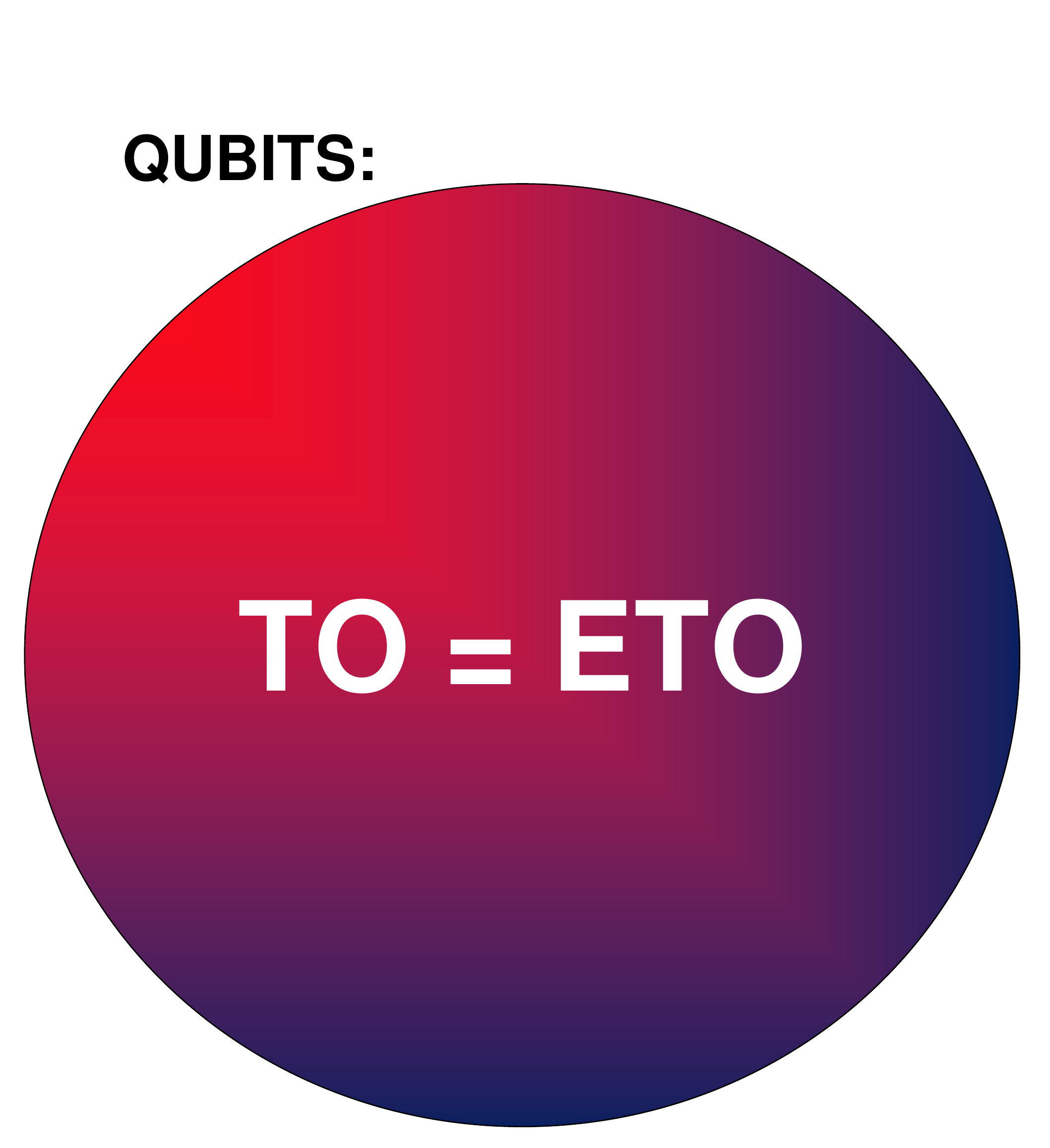} &
\includegraphics[width=0.20\textwidth, height=0.20\textheight]{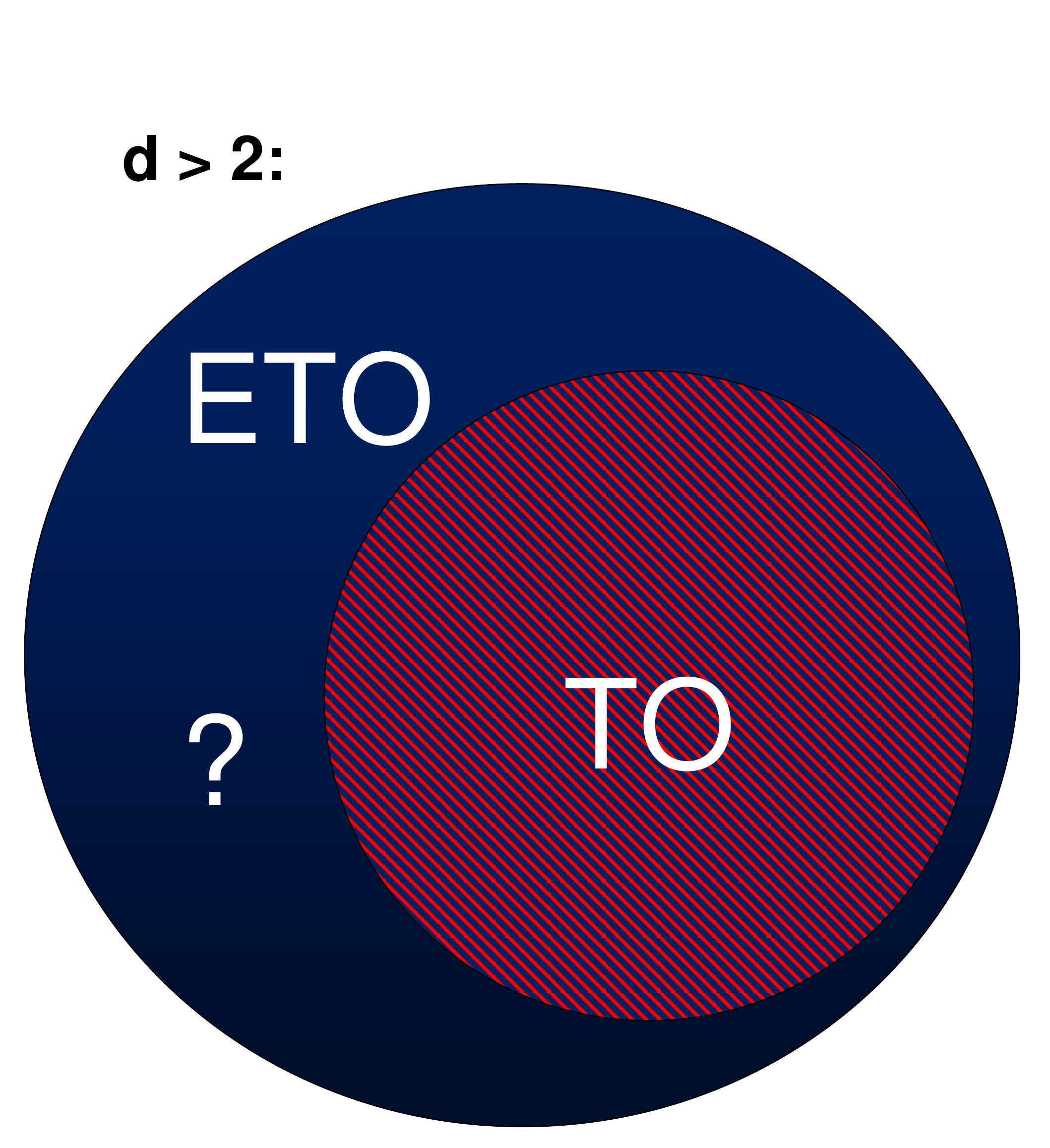}
\end{array}$
\end{center}
\caption{Comparison of Thermal Operations (TO) and Enhanced Thermal Operations (ETO) for qubits ($d=2$) (a) and $d>2$ (b). For $d=2$, they are equivalent, and from the Birkoff primitive given in \cite{2011arXiv1111.3882B}, TO can reproduce not only the extreme maps (when our bounds are saturated) but any other from ETO with an arbitrary precision. On the other hand, when $d>2$, ETO may be a wider class than Thermal Operations. However, we only have a counterexample for an exact qutrit state-to-state transitions, where we find a transition that can be realized by ETO, but not by TO. Studies of approximate states transitions are needed to verify the possible gap between TO and ETO.}
\label{fig:sets}
\end{figure}

{\it Discussion and open questions}
We study the limit of state-to-state transformations under Thermal Operations, focusing mostly on coherences and their preservation. We also introduce a new class of operations - Enhanced Thermal Operations and compared its power to TO. A natural research direction is to study whether they really outperform Thermal Operations if we are only concerned about approximate transformations.

It would be also interesting to check how state transitions look like if we add additional ancillas, and allow them to be returned in approximately the same state as before. Such catalytic thermal operations were studied in \cite{2013arXiv1305.5278B} and depending on the level of approximation, they effectively allow one to only approximately conserve  energy on the system. 

Finally, we have seen that the second laws we have introduced in the form of the \poscrit criteria are not strictly necessary and sufficient limitations on thermodynamical transformations. This likely means that there are more second laws which have to be satisfied. Finding them is an interesting open question.

{\it Acknowledgments}
We want to thank David Reeb for useful discussions. We would also like to thank Anna Studzi\'nska for preparing some of figures. J. O. is supported by the Royal Society. M. S. and P. \'C. are supported by the grant 2012/07/N/ST2/02873 from National Science Centre. M. H and P. \.C are also supported by the Seventh framework programme grant RAQUEL No 323970. M. H. is also supported by the NCN grant Maestro DEC-2011/02/A/ST2/00305. Research was (partially) completed while J.O., M. S. and P. \'C. were visiting the Institute for Mathematical Sciences, National University of Singapore in 2013, and (by all authors) during the program "Mathematical Challenges in Quantum Information" at the Isaac Newton Institute for Mathematical Sciences in the University of Cambridge (2013). Part of this work was done in National Quantum Information Centre of Gda\'nsk.

\widetext
\appendix
\renewcommand{\thesection}{\arabic{section}}
\renewcommand{\thesubsection}{\Roman{subsection}}
\renewcommand{\appendixname}{Supplementary Note}
\section{(Q)Bit of notation and assumptions}
\label{App:not}
Before proving and presenting in details the main results of our work, let us introduce some notation first. First, we will recall facts from \cite{Horodecki_2013thermo} to which we will add some new assumptions and modifications in the end.

Let us define $\ideta^X_E$ as a state of a system X proportional to
the projection on to a subspace of energy $E$ (according to the Hamiltonian $H_X$ on this system).
In particular, $\etaE$ is given by
\be
\etaE=g(E-E_S)^{-1}\sum_g |E-E_S,g\>_R\<E-E_S,g|
\ee
where $g=1,\ldots,g(E-E_S)$ are degeneracies, i.e. $\etaE$ is the maximally mixed state of the reservoir with
support on the subspace of energy $E_R=E-E_S$, $E_R$ are the energies of reservoir and $E_S$ of the system.
We shall also use notation $\eta_K=\id/ K$ where the identity acts on a $K$ dimensional space.

Let us note that the total space $\hcal_R\ot \hcal_S$ can be decomposed as follows
\be
\label{eq:Hoplus}
\hcal_R\ot \hcal_S= \bigoplus_E \left(\bigoplus_{E_S} \hcal^R_{E-E_S} \ot \hcal^S_{E_S} \right)
\ee

Consider an arbitrary state $\rho_{RS}$ which has
support within $\esmax\leq E\leq \ermax$. We can rewrite it as follows
\be
\label{proj}
\rho_{RS}=\sum_E \sum_\Delta P_E \rho_{RS} P_{E+\Delta}\\
\ee
Here $\Delta= -\esmax, \ldots, \esmax$.  The blocks
$P_E \rho_{RS} P_{E+\Delta}$ we can further divide into sub-blocks
\be
P_E \rho_{RS} P_{E+\Delta}=
\sum_{E_S\in I_\Delta} \id_R \ot P_{E_S}  P_E \rho_{RS} P_{E+\Delta} \id_R \ot P_{E_S+\Delta}
\ee
where $I_\Delta=\{0,\ldots, \esmax-\Delta\}$ for $\Delta \geq 0$ and
$I_\Delta=\{-\Delta,\ldots, \esmax\}$ for $\Delta \leq 0$.
The sub-blocks map the Hilbert space $\hcal^R_{E-E_S}\ot \hcal^S_{E_S+\Delta}$ onto
$\hcal^R_{E-E_S}\ot \hcal^S_{E_S} $

We can then extract the state $\rhos$
\be
\rhos=\sum_{E_S,E_S'} P_{E_S} \rhos P_{E_S'}
\ee
as follows:
\be
P_{E_S} \rhos P_{E_S'} = \sum_{E} \tr_{\hcal^R_{E-E_S}}
(P^R_{E-E_S} \ot P_{E_S}  P_E \rho_{RS} P_{E+E_S'-E_S} P^R_{E-E_S} \ot P_{E_S'})
\ee

Let us make some assumptions now.

We can assume that Hamiltonians of all systems of concern have minimal energy zero.
Let $\ermax$, and $\esmax$ be the largest energy of the heat the bath and system, respectively
(of course a typical heat the bath will have $\ermax=\infty$).

Our heat the bath will be large, while our resource states will be small.
This means that the system Hilbert space will be fixed, while the
energy of the heat the bath (and other relevant quantities such as size of degeneracies) will tend to infinity. It is quite an important assumption, because when the energy spectrum of the bath is not highly degenerated, then the set of Thermal Operations is very small and restricted.

\begin{remark}
In principle, Thermal Operations need not satisfy detailed balance; they should merely preserve the Gibbs state as a whole.
\end{remark}

The heat the bath is in a Gibbs state $\tau_R$ with inverse temperature $\beta$. Moreover there exists
set of energies $\ecalr$ such that
the state of the heat the bath occupies energies from $\ecalr$ with high probability, i.e.
for the projector $P_\ecalr$ onto the states with energies  $\ecalr$ we have
\be
\tr P_\ecalr \rhor\geq 1-\delta
\ee
and it  has the following properties:

\begin{enumerate}
\item \label{as1} The energies $E$ in $\ecalr$ are peaked around some mean value, i.e.
they satisfy $E\in \{\<E\>- O(\sqrt{\<E\>}),\ldots \<E\>+O(\sqrt{\<E\>})\}.$

\item For $E\in\ecalr$ the degeneracies $g_R(E)$ scale exponentially with $E$, i.e.
\be g_R(E) \geq e^{c E}, \ee where $c$ is a constant.

\item \label{as3} For any three energies  $E_R,E_S$ and $E_S'$  such that $E_R\in\ecalr$ and
$E_S$, $E_S'$ are arbitrary energies of the system, there exist $E'_R \in\ecalr$
such that $E_R +E_S = E_R'+E_S'$.

\item For $E\in\ecalr$ the degeneracies $g_R(E)$ satisfy $g_R(E-E_S) \approx g_R(E) e^{-\beta E_S}$,
or more precisely:
\be
\left|\frac{g_R(E)e^{-\beta E_S'}}{g_R(E-E_S)}-1\right|\leq \delta
\ee
for all energies $E_S$ of the system $S$.
\end{enumerate}
One can notice that a product $\tau_R^{\ot n}$ of many copies of independent
Gibbs states satisfies the above assumptions.

We then have:
\begin{theorem} \cite{Horodecki_2013thermo}
\label{thm:oplus_E}
We consider the set of energies
\be
\ecal=\{E:E-E_S\in \ecalr\}
\ee
where $\ecalr$ satisfies assumptions for heat the bath listed above. Then
\be
\forall{E\in \ecal} \quad
||\frac{1}{p_E}P_E \rhor\ot\rhos P_{E+\Delta} - \oplus_{E_S}\etaE\ot P_{E_S} \rhos P_{E_S+\Delta}||\leq 2\delta
\label{eq:oplus_E}
\ee
and
\be
\sum_{E\in\ecal} p_E\geq 1 -2\delta
\ee
where $p_E=\tr(P_E \rhor \ot \rhos)$.
\end{theorem}

All the above is sufficient, when one considers states that are diagonal in energy basis. To deal with coherences, our figure of merit, an additional assumption is needed.

Let us denote the minimal energy of the bath by $E^{R}_{\min}$ and the maximal one by $E^{R}_{\max}$.
\begin{enumerate}
\item
Let us define a new set of energies $\ecalr'$ which is the set $\ecalr$ where we removed all energies that are less (greater) than $E^{R}_{\min} (E^{R}_{\max}) \pm \max E_S$. It has the property that for the projector $P_\ecalr'$ onto the states with energies  $\ecalr'$ we have
\be
\tr P_\ecalr' \rhor\geq 1-\delta'
\ee
\end{enumerate}

Theorem \ref{thm:oplus_E} also holds for this assumption, the only difference is that $\delta$ needs to be replaced by $\delta'$.

\section{Thermal Operations}
Thermal Operations is the class of operations preassumes that there is a heat bath, described by a Gibbs state. So, it is worthwhile to emphasize what TO paradigm can and what it cannot explain. Clearly, the paradigm is not meant to answer how statistical mechanics emerges from reversible mechanics, i.e. how the probabilities emerge (apart from traditional coarse graining, there an interesting more recent approach is to use entanglement). What TO is very good to describe, is what can happen with system which can interact with the heat bath.  And the particular advantage of TO paradigm, is that it also treats the microscopic system, without any approximations.
\label{App:to}
Let us, one more time, present the facts and properties of Thermal Operations. In Thermal Operations one can
\begin{enumerate}
\item Bring in an arbitrary system in a Gibbs state with temperature T (free resource).
\item Remove (discard) any system.
\item Apply a unitary that commutes with the total Hamiltonian.
\end{enumerate}
Mathematically, the class of Thermal Operations on a system with $(\rho_S, H_S)$ can be viewed as
\begin{equation}
	(\rho_S, H_S) \mapsto (\tr_R[U \, (\rho_S \otimes \tau_R^{\ot n} ) \, U^{\dag}], H_S) = (\sigma_S, H_S) \label{CPTP}
\end{equation}
where $\tau_R = e^{- \beta H_R}/Z_R$ is the thermal state of the reservoir for some reservoir Hamiltonian $H_R$ at inverse temperature $\beta$, $\rho_S$ and $\sigma_S$ are some initial and final states from $H_S$. The class is generated by the unitaries $U$ (which act on the system, bath...) which obey the (strong) energy conservation condition
\begin{equation}
\label{super}
	[U, H_S + H_R + H_W] = 0,
\end{equation}
where $H_W$ is the term that a clock, a work system and other object under consideration besides the system and bath. Conservation of the total energy also implies the average energy conservation. In \cite{Horodecki_2013thermo}, the authors, provide an indication of how to generalize the above to the case of time-dependent system Hamiltonian, with the help of an auxiliary system $S'$.
{\textbf{Example:}}
In the case, when one has only the system and the heat bath, the total Hamiltonian is
\be
H_{SR} = H_S \ot I_R + I_S \ot H_R,
\ee
and from the energy conservation relation we have that
\be
[U, H_{SR}] = [U, H_S \ot I_R] + [U, I_S \ot H_R] = 0.
\ee
It means that to have the non-trivial dynamics and state-to-state transitions
\be [U, H_S \ot I_R] = -[U, I_S \ot H_R] \neq 0.
\ee
When both commutators are equal to zero, everything trivialize.

Let us recall the properties of Thermal Operations:
\begin{enumerate}
\item They are completely positive trace preserving maps (CPTP maps).
\item They preserve the Gibbs state.
\item During state-to-state transitions, the diagonal elements of an evolving state are not mixed with the off-diagonal ones, i.e., for any element $|i\>_S\<j|$ of a state $\rho_S$,
\be \label{general}
\Lambda |i\>_S\<j| = \sum_{kl: E_k-E_l=E_i-E_j} a_{kl}^{ij} |k\>_S\<l|,
\ee
where $\Lambda$ is Thermal Operations and $a_{kl}^{ij}$ are factors by which state elements are multiplied after an evolution.
\item For Hamiltonians having non-degenerated Bohr spectra, the off-diagonal elements are also not mixed between themselves, i.e., for an off-diagonal element $|i\>_S\<j|$ one gets
\be \Lambda (|i\>_S\<j|) = \lambda_{ij} |i\>_S\<j|,\ i \neq j \label{main11} \ee
$\lambda_{ij}$ are factors by which the off-diagonal elements are multiplied (damped) during the transition.
\end{enumerate}
{\it Proof of the Properties 1-4.}
The first property comes from Eq. \eqref{CPTP}, since Thermal Operations are implemented by unital dynamics. The second property comes straightforwardly from the commutation relation between energy-preserving unitary matrices and the total Hamiltonian, $\left[U, H\right]=0$ \cite{Horodecki_2013thermo}. We need to focus in details on the last two from the above list. We will prove these attributes of Thermal Operations using the formalism from Supplementary Note \ref{App:not}.

Let us consider an evolution (under energy preserving unitary operation U) of an off-diagonal $|i\>\<j|$ element of a quantum state that acts on the system $S$. Identifying the blocks of fixed energy $E$ and using Eq. \eqref{proj} we get:
\be \label{proof1}
\bigoplus_E U_E \rho_R \ot |i\>\<j|_S U^{\dagger} \bigoplus_{E'} U^{\dagger}_{E'}. \ee
Let us now the state $\rho_R$ in its energy-basis as
\be \rho_R = \sum_{E_R} p(E_R) |E_R\> \<E_R|, \ee
and insert it into Eq. \eqref{proof1} obtaining,
\be \label{st}
\bigoplus_E U_E \sum_{E_R} p(E_R) |E_R\> \<E_R| \ot |i\>\<j|_S U^{\dagger} \bigoplus_{E'} U^{\dagger}_{E'} = \sum_{E_R} p(E_R) \bigoplus_{E,E'} U_E |E_R\>|i\>_S\<E_R|\<j|_S U^{\dagger}_{E'}.
\ee
Let us examine now the action of $U$ on matrix elements of states. The only elements that are going to remain are that whose total energy of the system and bath is equal to $E$ ($E'$). It gives
\be \label{uact}
\begin{split}
&U_E|E_R\>|i\>_S = \sum_k \alpha_{E_R}^{ki} |E_R + E_i - E_k\> |k\>_S\delta_{E, \ E_R+E_i}, \\
&\<E_R|\<j|_S U^{\dagger}_{E'} = \sum_l {\alpha_{E_R}^{lj}}^{\star} \<E_R + E_j - E_l|\<l|_S \delta_{E', \ E_R+E_j}
\end{split}
\ee
where $\delta_{E,E_R+E_i}$ and $\delta_{E', \ E_R+E_j}$ are Kronecker deltas, $E_x$ is the energy of an element $|x\>$, and $\star$ stands for the complex conjugation. Inserting Eq. \eqref{uact} into Eq. \eqref{st} leads to
\be \label{ee1}
\sum_{E_R} p(E_R) \bigoplus_{E,E'} \sum_{kl}  \alpha_{E_R}^{ki} {\alpha_{E_R}^{lj}}^{\star} |E_R + E_i - E_k\> |k\>_S \<E_R + E_j - E_l|\<l|_S \delta_{E,\ E_R+E_i} \delta_{E',\ E_R+E_j}.
\ee
Using Kronecker deltas we get
\be \label{ee2}
\sum_{E_R} p(E_R) \sum_{kl}  \alpha_{E_R}^{ki} {\alpha_{E_R}^{lj}}^{\star} |E_R + E_i - E_k\> |k\>_S \<E_R + E_j - E_l|\<l|_S.
\ee
From Eq. \eqref{ee1}, precisely, from the Kronecker deltas, we get that $E'= E+E_j-E_i = E + \omega_{ij}$. From now on, we will denote $E_x-E_y=\omega_{xy}$, meaning $\omega_{xy}$ is the frequency between levels $x$ and $y$.

We can then rewrite Eq. \eqref{ee1} as
\be
\sum_{E_R} p(E_R) \left(\sum_{kl} \alpha_{E_R}^{ki} {\alpha_{E_R}^{lj}}^{\star} |E_R + \omega_{ki}\> \<E_R + \omega_{lj}| \otimes |k\>\<l|_S\right).
\ee
Applying the partial trave over the heat bath gives
\be \label{ufin}
\sum_{E_R} p(E_R) \left(\sum_{kl} \alpha_{E_R}^{ki} {\alpha_{E_R}^{lj}}^{\star} \<E_R + \omega_{lj}|E_R + \omega_{ki}\> |k\>\<l|_S\right).
\ee
To have non-zero scalar product, $\omega_{ki} = \omega_{lj}$. Keeping in mind that $\omega_{ki}=E_i-E_k$ and $\omega_{lj}=E_j-E_l$, the scalar product is non-zero, iff $E_k-E_l=E_i-E_j$. After the calculation, we get
\be \label{ufin1}
\sum_{kl: E_k-E_l=E_i-E_j} \left( \sum_{E_R} p(E_R) \alpha_{E_R}^{ki}{\alpha_{E_R}^{lj}}^{\star} \right) |k\>\<l|_S = \sum_{kl: E_k-E_l=E_i-E_j} a_{kl}^{ij} |k\>_S\<l|.
\ee
This proves Eq. \eqref{general}. But particulary, for systems having non-degenerate Bohr spectra, we have $E_i \neq E_j$ and $E_k\neq E_l$, then $E_k=E_i$ and $E_l=E_j$ (it corresponds to the situation when both $\omega$'s are equal to 0), which changes Eq. \eqref{ufin1} into
\be
\left(\sum_{E_R} p(E_R) \sum_{i,j} \alpha_{E_R}^i {\alpha_{E_R}^j}^{\star} \right)|i\> \<j|_S = \lambda_{ij}|i\> \<j|_S.
\ee
We get that our off-diagonal element is, after evolution, multiplied by $\Lambda_{ij}$ that depends only on $|i\> \<j|$ and not on any other off-diagonal element, which proves Eq. \eqref{main11}.

Property 3 implies the following
\begin{corollary}
For $\Lambda$ being Thermal Operations as in Eq. \eqref{CPTP}, one has $\left[\Lambda, \hat{H}\right] = 0$, where $\hat{H}(X) = [H, X]$.
\end{corollary}
{\it Proof.} For $\hat{H}$, we have that its eigenvalues are Bohr frequencies, and eigenspaces $P$ with respect to $\omega$ are given by $P_{\omega}:\ span\left\{|k\>\<l|, \ \omega = E_k-E_l \right\}$. From Property 3, we know that these eigenspaces are invariant for $\Lambda$ \cite{Alicki_book}.

As a remark, let us notice that one can derive laws of thermodynamics under Thermal Operations.  The commutation relation from Eq. \eqref{super} is equivalent to the first law (energy conservation)\cite{Horodecki_2013thermo}
and the unitarity (conservation of information) results in the second law(s) \cite{Horodecki_2013thermo,2013arXiv1305.5278B}. The third law can also be obtained in this framework\cite{arxiv.1412.3828}. One can show that the only allowed state in (1) that can be brought in for free is the equivalence class of Gibbs states at temperature $T$ \cite{2013arXiv1305.5278B}. Allowing any other state  would lead to the situation where there are simply no real limitations on possible transformations - every transformation is possible and there is no room for obtaining any bounds. This  can be considered as the zeroeth law which helps us to define the temperature\cite{2013arXiv1305.5278B}.
The assumption that a Gibbs state is the only possible free resource is crucial. Allowance of other states acting as a resource would lead to the situation where there are simply no real limitations on possible transformations - every transformation is possible and there is no room for obtaining any bounds \cite{2013arXiv1305.5278B}.
\section{Additional notation and assumptions for states with off-diagonal elements}
\label{App:add}
Let us try to sum up the above and adapt it to deal with the coherences.

We take a product quantum state of the Gibbs states from the heat bath and the system's state using Eq. \eqref{proj}. Since, we want to consider thermodynamical transitions, by means of Thermal Operations, we need to focus on energy preserving unitary transformations $U$ acting in blocks of fixed {\textit total} energy. We identify the blocks of fixed energies $E^i = E^j_S + E^{i,j}_{R}$ and due to the assumption \ref{as3} about heat the bath, we know that there always exist two different combinations of a sum of system and the bath energies that gives the same energy $E^i$, i.e., for the qubit with energies $E^0_S$ and $E^1_S$ one has $E^i = E^0_S + E^{i,0}_R = E^1_S + E^{i,1}_R$. The unitary transformation acts in these sub-blocks that have the same energy {\textit total} energy. Mathematically, we consider block-unitary transformations which can be written in the following form
\be
\label{gen_unit}
\mathbb{U}=\bigoplus_i U_i,
\ee
where each block acts on energy $E^i$.
Next we assume that an arbitrary unitary block $U_k$ from the above sum has the structure
\newcommand*{\tempb}{\multicolumn{1}{|c}{B_{k-1}}}
\newcommand*{\tempd}{\multicolumn{1}{|c}{D_{k-1}}}
\newcommand*{\tempbb}{\multicolumn{1}{|c}{B_{k}}}
\newcommand*{\tempdd}{\multicolumn{1}{|c}{D_{k}}}
\newcommand*{\tempo}{\multicolumn{1}{|c}{0}}
\newcommand*{\tempx}{\multicolumn{1}{|c}{\times}}
\be
\label{a8}
U_k=\left(
\begin{array}{cc}
A_{k} &\tempbb \\ \hline
C_{k} &\tempdd
\end{array}\right),
\ee
where submatrices $A_k$ and $D_k$ are square matrices of dimensions $d_{k} \times d_{k}$ and $d_{k-1} \times d_{k-1}$ respectively, while submatrices $B_k$ and $C_k$ are rectangular matrices of dimensions $d_{k} \times d_{k-1}$ and $d_{k-1} \times d_{k}$ respectively.

Then, we let this unitary acts on our state, obtaining a structure presented in Fig.~\ref{fig:diag1}. In the next sections, we show that from detail studies of the action of energy-preserving unitaries on states, the main results of our paper can be obtained.
\begin{figure}[ht]
\begin{center}
\includegraphics[scale=0.25]{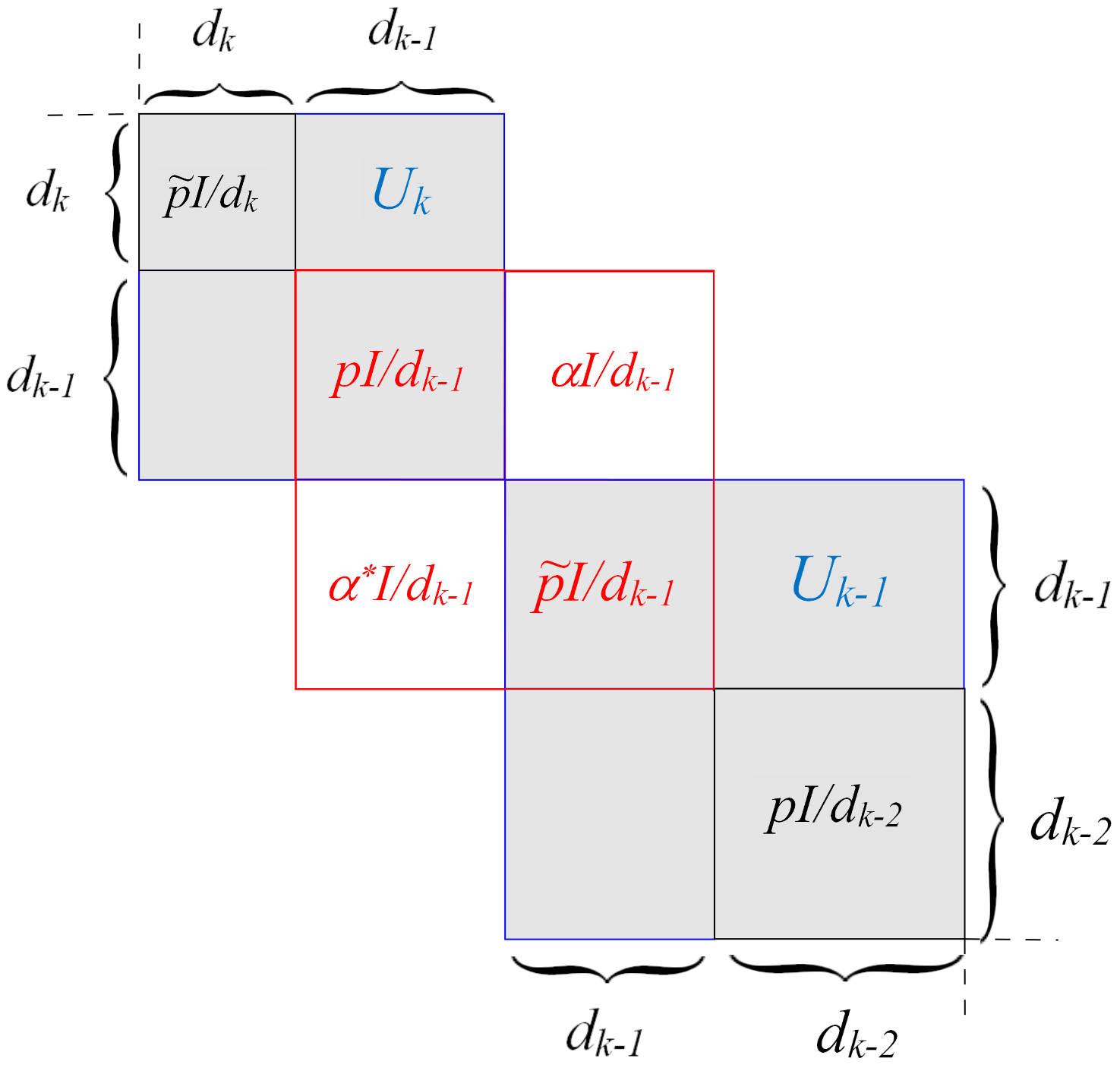}
\end{center}
\caption{\textbf{(Color online)} A qubit state $\rho_S = \usebox{\smlmat}$ with system energies $E^0_S$ and $E^1_S$, projected onto energy blocks $E^i_R$ of heat the bath with degeneracies of dimensions $d_i$. Darker squares correspond to blocks of set energies $E^i = E^j_S + E^i_R$, where a unitary $\mathbb{U}=\bigoplus_i U_i$ acts to transform $\rho_S$ into other state.}
\label{fig:diag1}
\end{figure}

We need to fulfill our assumptions, so to model the energies in our setup, (for an arbitrary dimension of a state) we can use the multinomial distribution, so the assumption \ref{as1} is obeyed. We also need to work in a regime, where dimensions of degeneracies, in a region of energy distribution, are non-decreasing.
Moreover, blocks of unitaries connect blocks of different bath energies, so to ensure we are in a proper regime (that fulfills all our assumptions), we first have to make a cut on system and the bath energies, using the Chernoff bound \cite{Mitzenmacher:2005:PCR:1076315}, so they fit blocks of unitary, and the assumption \ref{as3} is followed, and we have a non decreasing order of dimensions of degeneracies. In other words, in a general situation, we can identify 3 different steps:

\begin{itemize}
	\item First, we make a cut, so the law about degeneracies is fulfilled.
	\item We move an energy by a value $\operatorname{max}E^i_S$ into the direction of mean energy, where weight of energies are big.
	\item We want to have a unitary operation that acts on fully on blocks of set energies, so we take the projection from this regime, where all components of set energies (sum of systems and the bath energies), are from the region where to which we cut our energy distribution. Sometimes, to fulfill this, we need to include some energies which are already outside the cut area (by Chernoff bounds), but their weights are small, so we can take them to have all components of set energies.
\end{itemize}

\section{Seconds laws for coherences}
\label{App:sts}
Since we know that what happens of coherences can be separated on that what happens on a diagonal of a state, we can divide the full algorithm for state-to-state transitions into two steps
\subsection{Thermomajorization}
There exist a necessary and sufficient method, derived in \cite{Horodecki_2013thermo}, called thermomajorization, to check whether a transition between states ($\rho_S$, $H_S$) $\rightarrow$ ($\sigma_S$, $H_S$) is possible, when both states
commute with the Hamiltonian $H$ of the total system, which means they are diagonal in the energy-basis. It is based
on the majorization condition for state transformations, which is
a necessary and sufficient condition for state transformations
under permutation maps. The very brief idea is to write the
eigenvalues of the state and heat bath, in terms of eigenvalues of only the state, order them in a nonincreasing way and compare integrals of the connected functions (some monotones).

To present the details, we will use the original formulation taken from \cite{Horodecki_2013thermo}.

Let $p_{E_S,g}$  be eigenvalues of $\rho$ and $q_{E_S,g}$ be eigenvalues of $\sigma$.
The state
$P_E \rhor\ot \rhos P_E$ after normalization  is close to the state having  the following eigenvalues:
\be
e^{\beta E_S}\frac{p(E_S,g)}{g_R(E)}
\label{eq:eig}
\ee
with multiplicity $g_R(E) e^{-\beta E_S}$, where $E_S$ runs over  all
energies of the system, and $g$ runs over degeneracies.
Similarly, $P_E \rhor\ot \sigma_S P_E$ has eigenvalues $e^{\beta E_S}\frac{q(E_S,g)}{g_R(E)}$
with the same multiplicity.

The eigenvalues are very small, and they are collected in groups, where they are the same,
hence the majorization amounts to comparing integrals.
If one puts eigenvalues into decreasing order, one obtains a stair-case like function,
and majorization in this limit will be to compare the integrated functions (which are then
piece-wise linear functions).

To see how it works, we need to put the eigenvalues in nonincreasing order.
The ordering is determined by the ordering of the quantities $e^{\beta E_S}p_{E_S,g}$.
This determines the order of $p(E_S,g)$ (which in general will not be in decreasing order anymore).
We shall denote such ordered probabilities as $p_i$,
and the associated energy of the eigenstate as $E_i$.
E.g. $p_1$ is equal to the $p(E_S,g)$ such that $e^{\beta E_S}p(E_S,g)$ is the largest.
Note that for fixed $E_S$ the order is the same as the order of $p({E_S,g})$,
while for different $E_S$ it is altered by the Gibbs factor. We do the same for $\sigma$,
which results in $q_i$.

The eigenvalues are thus ordered by taking into account Gibbs weights:
\be
\underbrace{\frac{p_1 e^{\beta E_1}}{{d_E}}}_{{\text{multiplicity} \approx d_E e^{-\beta E_1}}}
\geq
\underbrace{\frac{p_2 e^{\beta E_2}}{{d_E}}}_{{\text{multiplicity} \approx d_E e^{-\beta E_2}}}
\geq \ldots
\label{eq:beta-order2}
\ee
where $d_E$ is a shorthand for $g_R(E)$.
We shall now ascribe to vector $\{p_i\}$ a function mapping interval $[0,Z]$ into itself.
On the $y$ axis, we put subsequent sums $\sum_{i=1}^l p_i$, $l=1,\ldots ,d$ where
$d$ is the number  of all probabilities,
and on the $x$ axis, we put sums $\sum_{i=1}^l e^{-\beta E_i}$, with the final point being at $x=Z$.
This gives $d+1$ pairs: $(0,0), (p_1, e^{-\beta E_1}), (p_1+p_2, e^{-\beta E_1}+e^{-\beta E_2}), \ldots,
(Z,1)$. We join the points, and it will gives us a graph of a function, $f_p(x)$.
It is easy to see, that in the limit of large $g_R(E)$, the eigenvalues of $\rho$ majorize eigenvalues of $\sigma$
if and only if $f_p(x)\geq f_q(x)$ for all $x\in[0,Z]$. The described scheme is also presented in Fig.
\ref{fig:thermomaj1}.
\begin{figure}[h]
  \includegraphics[width=0.45\textwidth]{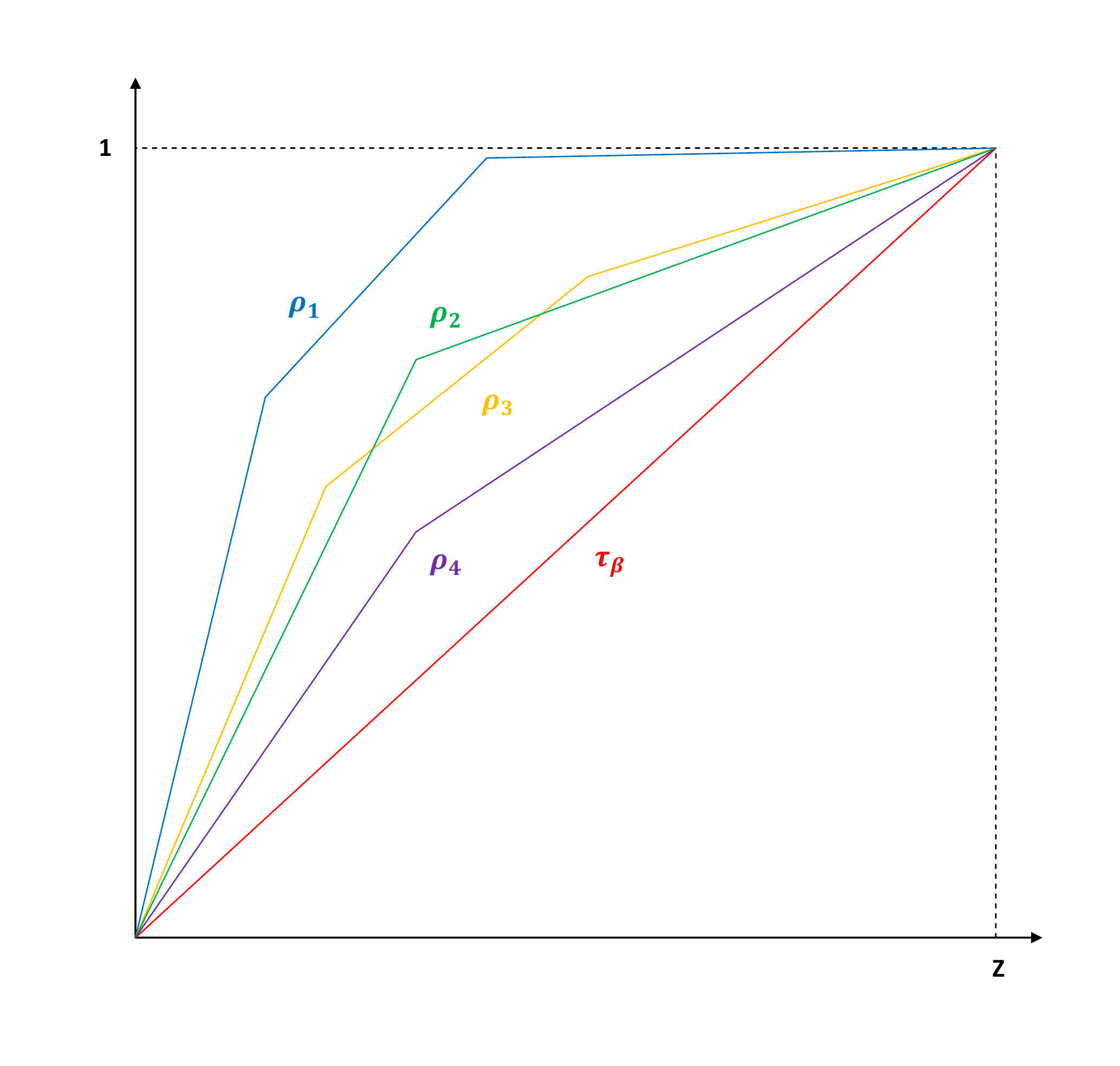}
  \caption{\textbf{(Color online)} The thermomajorization criteria is as follows: Consider probabilities $p(E,g)$ of the initial system $\rho_i$ to be in the $g$'th state of energy $E$. Now let us put $ p(E,g)\e^{\beta E}$ in decreasing order $p(E_1,g_1)\e^{\beta E_1}  \geq  p(E_2,g_2)\e^{\beta E_2}\geq p(E_3,g_3)\e^{\beta E_3} ...$, we say that the eigenvalues are $\beta$-ordered.  We can do the same for other system $\rho_j$
i.e. $\e^{\beta E^{'}_{1}} q(E^{'}_{1},g^{'}_{1}) \geq \e^{\beta E^{'}_{2}} q(E^{'}_{2},g^{'}_{2})\geq\e^{\beta E^{'}_{3}} q(E^{'}_{3},g^{'}_{3})...$.  Then the condition which determines whether
we can transform $\rho_i$ into $\rho_j$ is depicted in the above figure.  Namely, for any state, we construct a curve with points $k$ given
by $\{\sum e^{-\beta E_i}/Z,\sum_i^k p_i \}$. Then a thermodynamical transition from $\rho_i$ to $\rho_j$ is possible if and only if, the curve of $\rho_i$
lies above the curve of $\rho_j$. The Gibbs state $\tau_{\beta}$ always has the lowest possible cumulative plot. One can make a previously impossible transition possible by adding work $E$ which will scale each point by an amount $e^{-\beta E}$ horizontally. Based on this, in the figure, one can go from $\rho_1$ to any other state, but from $\rho_2$ it is possible to reach $\rho_4$ and $\tau_{\beta}$ only, and so on.}
\label{fig:thermomaj1}
\end{figure}

\subsection{Limitations for processing of coherences}
Let us recall the bounds for coherences:
\begin{proposition}
\label{p1}
When a transformation between two d-level system, initial - $\rho_S$ and final - $\sigma_S$ occurs, by means of Thermal Operations, the bounds for coherences transport come from the positivity of the Choi map that is associated with the energy preserving dynamics:
\be \label{bound} \begin{bmatrix}
p(0 \rightarrow 0) & \alpha_{11} & \ldots & \alpha_{1n}\\
\alpha_{21} & p(1 \rightarrow 1) & \ldots & \alpha_{2n}\\
\vdots & \vdots & \ddots & \vdots \\
\alpha_{n1} & \alpha_{n2} & \ldots & p(n \rightarrow n)
\end{bmatrix}, \ee
where $\alpha_{ij}$ are factors by which the off-diagonal elements are multiplied during the transition, and $p(i \rightarrow i)$ are probabilities of staying in the same energy-level; ($p(i \rightarrow j)$ is a transition
probability and $p(i)$ is a probability of occupying an energy state $i$)
\end{proposition}
It shows that the bounds for the transport of off-diagonal elements come from the minors of the matrix from Eq. \eqref{bound}. For example, for qubits, we have 2 minors to consider, one trivial (that probabilities $p(i \rightarrow i)$ should be non-negative) and the one that really damps the coherences
\be
p(0 \rightarrow 0)p(1 \rightarrow 1) \leq \alpha^2.
\label{damp} \ee
Let us examine the qubit's case in greater details.
\subsection{Qubit example}
For qubits, there is only one bound for coherences transport and the inequality is tight, which implies that the criteria for qubits is necessary and sufficient.
To see that, consider two qubit states $\rho_S = \begin{bmatrix}
p & \alpha\\
\alpha^{*} & 1-p
\end{bmatrix}$ and $\sigma_S = \begin{bmatrix}
q & \chi\\
\chi^{*} & 1-q
\end{bmatrix}$, written in the energy eigenbasis, on a system with Hamiltonian $H_S$, where $*$ stands for complex conjugation. We know that the evolution of diagonal elements can be separated from off-diagonal ones, so for diagonal elements one uses thermomajorization obtaining four different situations, depending on the diagonal input and output and energies of the system:
\be \frac{f}{1-f} \geq \E^{\beta (E_1-E_0)}, \label{crit11} \ee
or
\be \frac{f}{1-f} \leq \E^{\beta (E_1-E_0)}, \label{crit22} \ee
where $f$ is $p$ or $q$, $E_i$ is an energy of the system (of levels 1 and 0). The four cases are illustrated in Fig. \ref{fig:beta}
\begin{figure}[h]
  \includegraphics[width=0.45\textwidth]{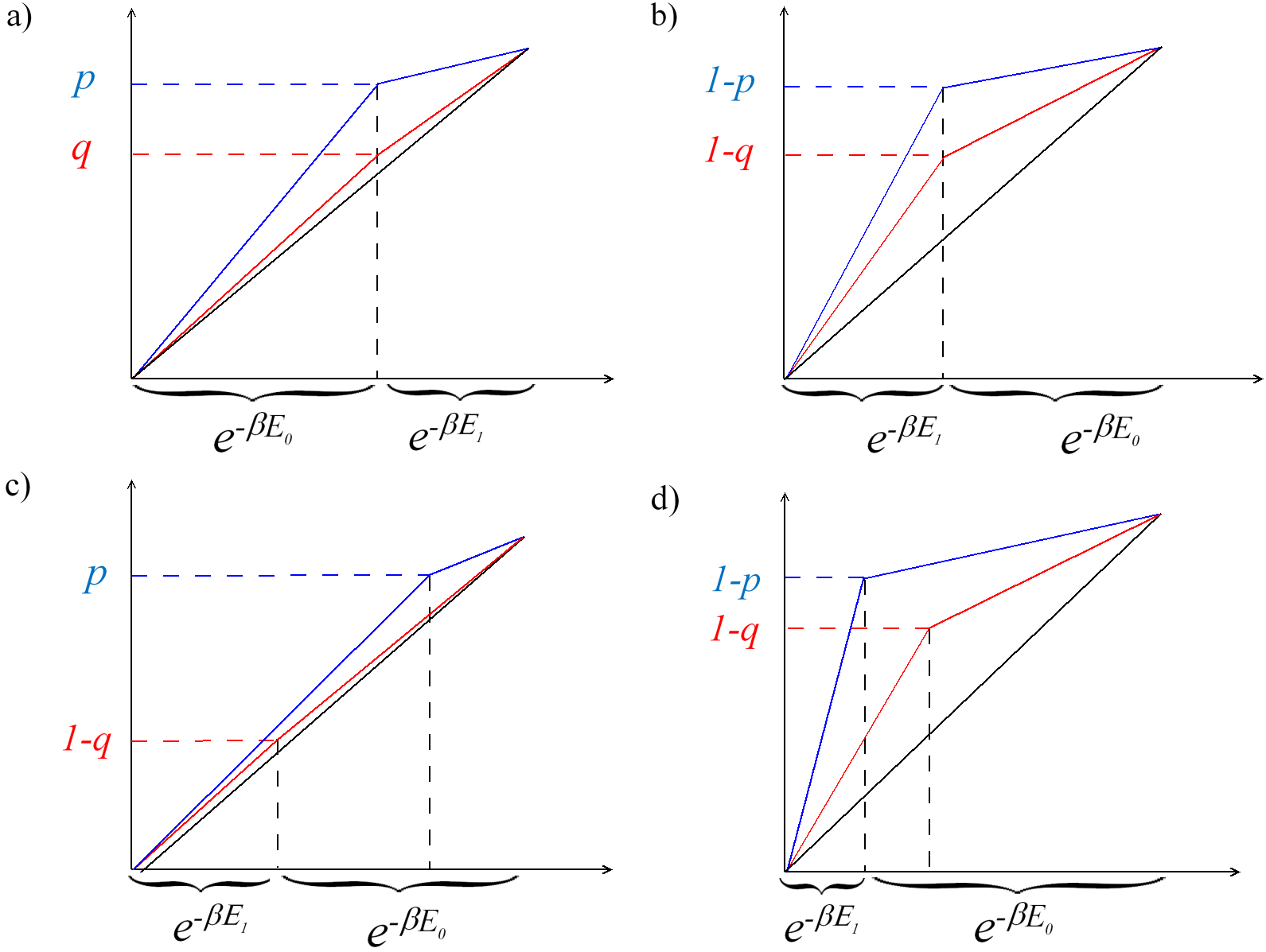}
  \caption{{\it Different $\beta$-orders}. Four cases that follows from different $\beta$-order are presented. In $a)$ and $b)$ there are the same $\beta$-order, in $c)$ and $d)$ different. Also in $a)$ both states corresponding to the curves are less excited than Gibbs state, in $b)$ more excited, in $c)$ upper state - less excited, lower - more and in $d)$ reverse. The lowest possible line corresponds to the Gibbs state.}
\label{fig:beta}
\end{figure}

Then, the state-to-state transformation $(\rho, H_S) \rightarrow (\sigma, H_S)$ by means of Thermal Operations is possible if and only if:\\
\begin{enumerate}
\item Diagonal elements - implied by thermomajorization.\\
$\begin{array}{c}
\text{a: When both $p$ and $q$ fulfills \ref{crit11} then
$p \geq q$} \\
\text{b: when both $p$ and $q$ fulfills \ref{crit22} then
$p \leq q$} \\
\text{c: when $p$ fulfills \ref{crit11} and $q$ \ref{crit22} then
$\frac{q}{1-p} \geq \frac{r}{1-r}$}\\
\text{d: when $p$ fulfills \ref{crit22} and $q$ \ref{crit11} then
$\frac{p}{1-q} \leq \frac{r}{1-r}$}.
\end{array}$
 \item Off-diagonal elements - fundamental bound for all cases from Eq. \eqref{damp}. \\
  \begin{equation}
   |\chi| \leq  |\alpha| \kappa,
   \label{case5}
   \end{equation}
\end{enumerate}
where $\kappa =  \frac{\sqrt{(q-\widetilde{p}\E^{\beta \Delta E})(p-\widetilde{q}\E^{\beta \Delta E})}}{\left| p-\widetilde{p}\E^{\beta \Delta E} \right|}$, $\widetilde{q} = 1 - q$, $\widetilde{p} = 1- p$, $\E^{\pm\beta \Delta E} = \E^{\pm\beta (E_j-E_i)}$ with $E_i$ being energy of the system and $\beta$ inverse temperature $\beta = \frac{1}{kT}$. Notice that $\kappa \geq 0$, where it is equal to zero iff we consider the excited state or the ground state, i.e., where $p=0$ or $p=1$, but these states have no coherences at all, so we can conclude that qubit optimal processes are not able to destroy coherences completely (even if repreated many times).    
{\it Remark.} Since phases commute with the total Hamiltonian of our setup, we can restrict our attention only to moduli of the coherences.

\subsection{Proof of the optimal coherence transport for qubits}
It this subsection, we show that the qubits state-to-state transitions can be always implement by such an energy preserving U that the coherences are damped in a most friendly way for them, ie, that it is given by the equality from Eq. \eqref{damp}.
\begin{proposition}
For qubits, the bound from Proposition \ref{bound}, which gives Eq. \eqref{case5}, is tight.
\end{proposition}
Let us take two blocks $D_{k}$ and $A_{k-1}$ of a dimension $d_{k-1}$,  like in Figure~\ref{fig:diag11}. We assume here that matrices $D_{k}$ and $A_{k-1}$ have the diagonal form
\begin{figure}[ht]
\begin{center}
\includegraphics[scale=0.25]{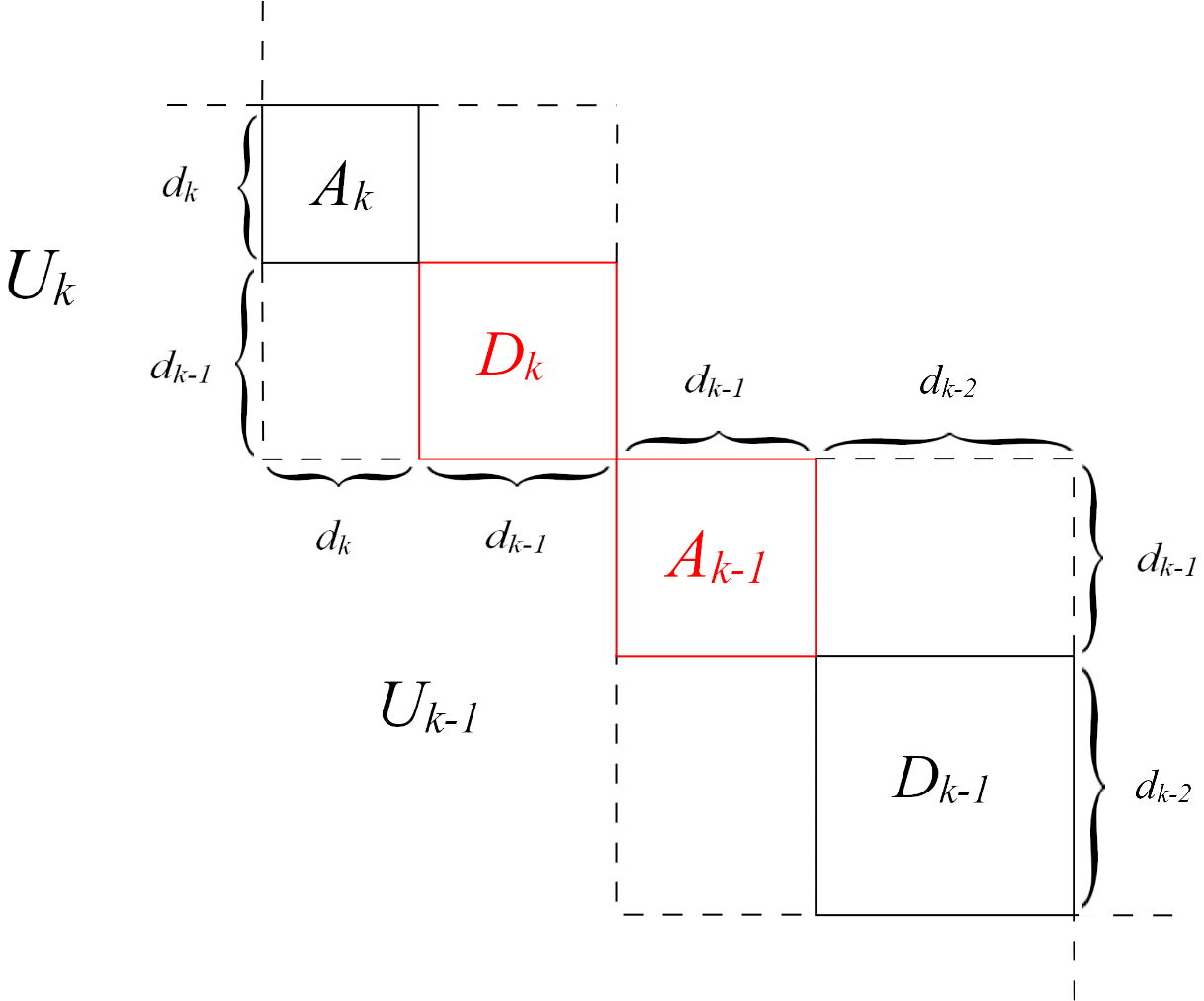}
\end{center}
\caption{\textbf{(Color online)} We consider two unitary blocks $U_{k}$ and $U_{k-1}$ and their subblocks $D_{k}$ and $A_{k-1}$ (red color) of dimensions $d_{k-1} \times d_{k-1}$. Our goal is to find maximal value of $\tr \left[ D_{k}A_{k-1}^{\dagger}\right]$, when both matrices are diagonal and satisfy constraints $\tr \left[ D_{k}D_{k}^{\dagger}\right]=d_{k-1}p(0 \rightarrow 0)$, $\tr \left[A_{k-1}A_{k-1}^{\dagger}\right]=d_{k-1}p(1 \rightarrow 1)$. By $p(0 \rightarrow 0)$ and $p(1 \rightarrow 1)$ we denote probability transitions between energy levels $0 \rightarrow 0$ and $1 \rightarrow 1$ respectively.}
\label{fig:diag11}
\end{figure}
and satisfy constraints which follow from their trace norms
\be
\label{a1}
\tr \left[ D_{k}D_{k}^{\dagger}\right]=\sum_{i=1}^{d_{k-1}}z_i^2=d_{k-1}p(0 \rightarrow 0),\qquad \tr \left[A_{k-1}A_{k-1}^{\dagger}\right]=\sum_{i=1}^{d_{k-1}}x_i^2=d_{k-1} p(1 \rightarrow 1).
\ee
We would like to know the maximal value of
\be
\label{a2}
\tr \left[D_{k}A_{k-1}^{\dagger}\right]=\sum_{i=1}^{d_{k-1}}z_ix_i
\ee
with constraints~\eqref{a1}. We can treat diagonals of matrices of $D_{k}$ and $A_{k-1}$ like a vectors of length $d_{k-1}$, so
\be
\label{a3}
|\psi\>=|z_1,\ldots,z_{d_{k-1}}\>,\quad |\varphi\>=|x_1,\ldots,x_{d_{k-1}}\>
\ee
with squared norms
\be
\label{a4}
|| \psi ||^2=d_{k-1}p(0 \rightarrow 0),\qquad || \varphi ||^2=d_{k-1}p(1 \rightarrow 1)
\ee
and equation~\eqref{a2} reads $\tr \left[ D_{k} A_{k-1}^{\dagger}\right]=\<\psi |\varphi\>$. Thanks to this,  maximization problem of $\tr \left[ D_{k} A_{k-1}^{\dagger}\right]$ we can reformulate to $\max \<\psi | \varphi\>$ with constraints~\eqref{a4}.\\ In the first step we change coordinates in the following way
\be
\label{a5}
\begin{split}
\sum_{i=1}^{d_{k-1}}z_i^2&=d_{k-1}p(0 \rightarrow 0) \rightarrow  \sum_{i=1}^{d_{k-1}}\bar{z}_i^2=1,\\
\sum_{i=1}^{d_{k-1}}x_i^2&=d_{k-1}p(1 \rightarrow 1) \rightarrow  \sum_{i=1}^{d_{k-1}}\bar{x}_i^2=1,
\end{split}
\ee
so  we have $\lambda_1 z_i=\bar{z}_i$ and $\lambda_2 x_i=\bar{x}_i$, where $\lambda_{1(2)}$ are some numbers. Using norm invariance we calculate that
\be
\label{a6}
\lambda_1=\frac{1}{\sqrt{d_{k-1}p(0 \rightarrow 0)}},\qquad \lambda_2=\frac{1}{\sqrt{d_{k-1}p(1 \rightarrow 1)}}.
\ee
Finally we can write
\be
\label{a7}
\begin{split}
&\max \<\psi|\varphi\>=\max \left(\sum_{i=1}^{d_{k-1}}z_ix_i\right)=\frac{1}{\lambda_1\lambda_2}\max\left(\sum_{i=1}^{d_{k-1}}\bar{z}_i\bar{x}_i\right)=\\
&=d_{k-1}\sqrt{p(0 \rightarrow 0)p(1 \rightarrow 1)}=\sqrt{\tr \left[ D_{k}D_{k}^{\dagger} \right]}\sqrt{\tr \left[ A_{k-1}A_{k-1}^{\dagger} \right]}.
\end{split}
\ee
From Eq.~\eqref{a7} we see that the maximum value of $\tr \left[D_{k} A_{k-1}^{\dagger}\right]$ is equal to the product $\sqrt{\tr \left[ D_{k}D_{k}^{\dagger}\right]}\sqrt{\tr \left[ A_{k-1}A_{k-1}^{\dagger}\right]}=d_{k-1}\sqrt{p(0 \rightarrow 0)p(1 \rightarrow 1)}$.  One can see that in previous calculations we did not assume nothing about exact values of $x_i$ and $z_i$, where $1\leq i \leq d_{k-1}$. We see that we saturate the inequality~\eqref{a7} when matrices $D_{k}$ and $A_{k-1}$ (respectively vectors $|\psi \>$, $|\varphi \>$) are parallel.

It seems that for an arbitrary states transformation, there are many unitaries that realize some transition. We need the one that, which has the largest probabilities $p(i \rightarrow i)$, as this would improve the transfer of coherences and possibly saturate bounds from Eq.~\eqref{a7}. Then we can ask: Are we able to construct a unitary transformation $\mathbb{U}$ where every blocks $D_{k}$ and $A_{k-1}$ saturates the equality from Eq.~\eqref{a7}? The answer for this question in the case of a two-level system is positive and the construction is like in Figure~\ref{fig:diag2}. One can note that that there is also another possibility of constructing a block-unitary transformation $\mathbb{U}$ which saturates the inequality~\eqref{a7} and does not have the "brute" form, i.e. we do not put only zeros and ones on the diagonal in the block $A_n$. It can be checked that filling the first block (left, upper corner) with $\operatorname{diag}A_n=\{\sqrt{p(0 \rightarrow 0)},\ldots,\sqrt{p(1 \rightarrow 1)}\}$, subsequently filling the next block with $\sqrt{p(1 \rightarrow 1)}$ and 1, and then performing the optimalization over the next possible block (first of the higher energy) also leads to saturation of the Schwartz inequality.
\begin{figure}[ht]
\begin{center}
\includegraphics[scale=0.25]{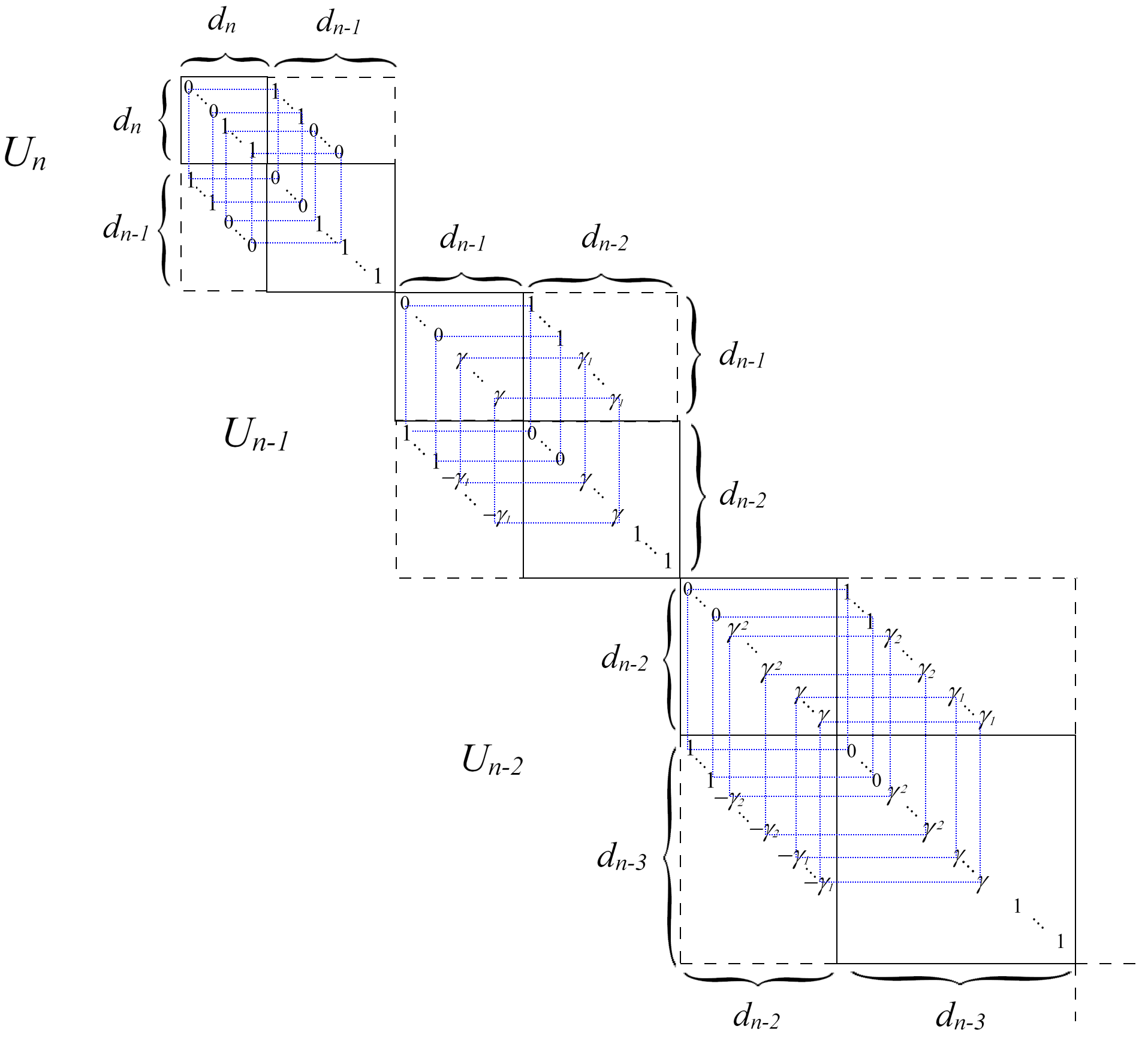}
\end{center}
\caption{\textbf{(Color online)} In this figure we present an optimal block-unitary transformation. We start our construction from unitary block $U_n$. Firstly on diagonal of submatrix $A_n$ we put some number of zeros and ones to satisfy constraint $\tr \left[ A_nA_n^{\dagger}\right]=d_np(1 \rightarrow 1)$. The numbers of zeros and ones tell us how much energy levels we want to move and leave untouched. On the off-diagonal blocks of $U_n$ we put ones to complete full block to unitary. Secondly, diagonal of submatrix $D_n$ we choose in such a way to obtain it parallel to diagonal of $A_n$, i.e. we rewrite diagonal of $A_n$ and put ones on the tail (we do not have more energy levels to move). The number of ones on diagonal of $D_n$ is equal to $d_np(1 \rightarrow 1)+(d_{n-1}-d_n)=d_{n-1}p(0 \rightarrow 0)$. Now let us take submatrix $A_{n-1}$ in $U_{n-1}$. Diagonal of $A_{n-1}$ has to be parallel to diagonal of $D_n$ and of course satisfy condition $l_n \gamma^2=d_{n-1}p(1 \rightarrow 1)$, where $l_n=d_{n-1}p(0 \rightarrow 0)$ is the number of ones in the submatrix $D_n$, so we have to choose $\gamma=\sqrt{\frac{p(1 \rightarrow 1)}{p(0 \rightarrow 0)}}$. We continue this procedure for the rest blocks in our unitary transformation. One can see that this construction is valid only in the regime $p(1 \rightarrow 1)< p(0 \rightarrow 0)$ (see, Lemma \ref{lemma1}). In this figure $\gamma_1=\sqrt{1-\gamma}$ and $\gamma_2=\sqrt{1-\gamma^2}$.}
\label{fig:diag2}
\end{figure}

\begin{figure}[ht]
\begin{center}
\includegraphics[scale=0.25]{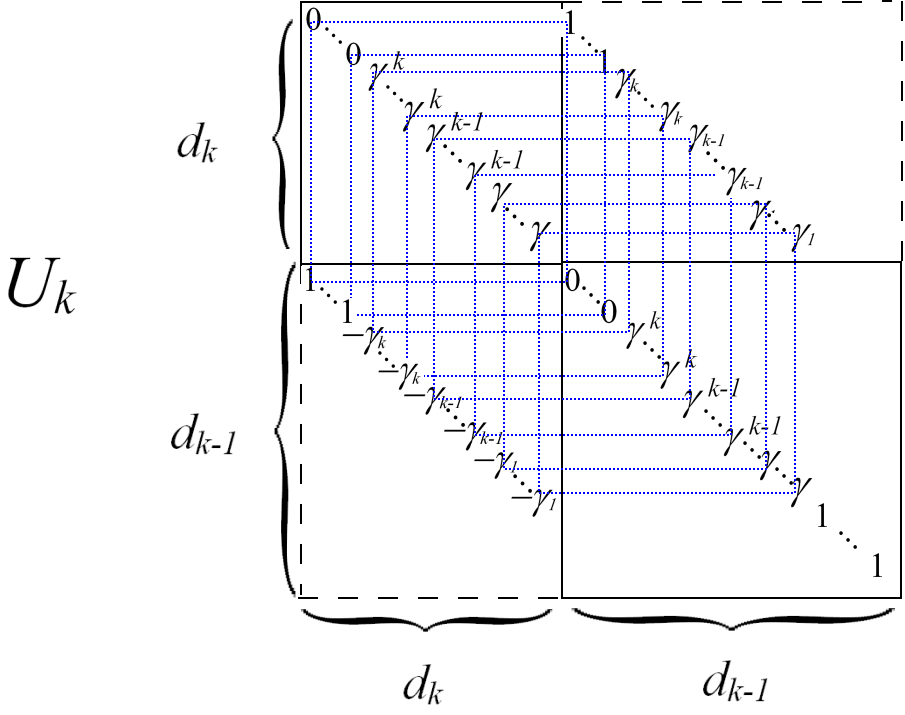}
\end{center}
\caption{\textbf{(Color online)} In this figure we present an arbitrary block $U_{k}$ of a general block-unitary transformation from Fig.~\ref{fig:diag2}. Let us consider a submatrix $D_k$. The number of zeros is equal to the numbers of zeros in our starting point - submatrix $A_n$. Then number of $\gamma^k$ is equal to $l_n$, number of $\gamma^{k-1}$ is equal to $d_{n-2}-d_{n-1}$ and finally number of $1'$s in the tail is equal to $d_{k-1}-d_{k}$. We define $\gamma_i$ as $\sqrt{1-\gamma^i}$. }
\label{fig:diag3}
\end{figure}

We need to check what is the relation between $p(0 \rightarrow 0)$ $p(1 \rightarrow 1)$. It is presented in the following lemma.
\begin{lemma}
\label{lemma1}
For the qubit case one has:
\be p(0 \rightarrow 0) > p(1 \rightarrow 1), \ee
and
\be p(1 \rightarrow 0) > p(0 \rightarrow 1). \ee
\end{lemma}
\begin{proof}
From the preservation of the Gibbs state one has
\be
\frac{p(i \rightarrow j)}{p(j \rightarrow i)} = e^{-\beta (E_j-E_i)} < 1,
\ee
which immediately tells us that $p(1 \rightarrow 0) > p(0 \rightarrow 1)$. Now, using a relation between $p(0 \rightarrow 0)$ and $p(1 \rightarrow 1)$: $d_{k}p(1 \rightarrow 1) + (d_{k-1}-d_k) = d_{k-1}p(0 \rightarrow 0)$, where $d_i$ are dimensions of degeneracies of the bath, we have that
\be
\begin{split}
&d_{k-1}p(0 \rightarrow 0)-d_{k-1} = d_{k}p(1 \rightarrow 1)- d_{k}p, \\
&d_{k-1}(p(0 \rightarrow 0)-1) = d_k(p(1 \rightarrow 1)-1),\\
&\frac{d_{k-1}}{d_k}=\frac{p(1 \rightarrow 1 )-1}{p(0 \rightarrow 0)-1}.
\end{split}
\ee
We know that Thermal Operations preserve the Gibbs state which equivalently means that $\frac{d_{k-1}}{d_k} > 1$ which gives $p(0 \rightarrow 0) > p(1 \rightarrow 1)$ q.e.d.
\end{proof}

The above lemma is used in the construction of the optimal unitary for coherences transport that is presented in Fig. \ref{fig:diag2}.

\subsection{Full characterization of qubits state-to-state transitions: classical channel analogy}
\label{sec:app}
In this section we give formulas which allow us to write probabilities $p(i \rightarrow j)$ of the level transitions, in terms of initial and final state elements of our case-study example, written in the energy-basis, $p,q, 1-p,1-q$, using an analogy of the classical channel (we have a channel that preserves the Gibbs state on a diagonal of input states, and some tranformation/damping of coherences, which comes from the energy conservation relation). From the positivity of this channel, we get that the Gibbs state needs to be preserved and the coherences need to be damped. To do this, let us consider a one block of the qubit in the initial state (see, Fig.~\ref{fig:diag1}.
Due to the transformation $\mathbb{U}$ we obtain as a result, an output state in the form
\be
\label{c2}
\sigma=\begin{pmatrix} \mathcal{C}\frac{\widetilde{p}}{d_{k}}+\mathcal{D}\frac{p}{d_{k-1}} & \times \\ \times & \mathcal{A}\frac{\widetilde{p}}{d_{k-1}}+\mathcal{B}\frac{p}{d_{k-2}} \end{pmatrix},
\ee
where $\mathcal{A}=A_{k-1}A_{k-1}^{\dagger}$, $\mathcal{B}=B_{k-1}B_{k-1}^{\dagger}$, $\mathcal{C}=C_{k}C_{k}^{\dagger}$, $\mathcal{D}=D_{k}D_{k}^{\dagger}$, and $\widetilde{p}+p=1, \ \widetilde{q}+q=1 $. Thanks to this and equation~\ref{c2} we can write
\be
\label{c3}
\begin{split}
&\frac{\widetilde{p}}{d_{k}} \tr \mathcal{C}+\frac{p}{d_{k-1}}\tr \mathcal{D}=q,\\
&\frac{\widetilde{p}}{d_{k-1}} \tr \mathcal{A}+\frac{p}{d_{k-2}}\tr \mathcal{B}=\widetilde{q}.
\end{split}
\ee
 Now we can rewrite this in terms of probabilities $p_{ij}$ using the preservation of the Gibbs state relation obtaining
\be
\label{channel}
\begin{split}
&q = p(0 \rightarrow 0) + \widetilde{p} \E^{\beta \Delta E} p(0 \rightarrow 1), \\
&\widetilde{q} = p\E^{-\beta \Delta E}p(1 \rightarrow 0) + \widetilde{p} p(1 \rightarrow 1).
\end{split}
\ee
Because of the constraints, coming from the unitarity of our block matrix, $p(1 \rightarrow 1)+p(1 \rightarrow 0)=1$ and $p(0 \rightarrow 0)+p(0 \rightarrow 1)=1$ we can easily express probabilities $p(i \rightarrow j)$ in terms of $\widetilde{p},p,\widetilde{q},q$
\be
\label{c4}
\begin{split}
p(0 \rightarrow 0)&=\frac{1}{d_{k-1}}\tr \mathcal{D}=\frac{q-\widetilde{p}\E^{\beta \Delta E}}{p-\widetilde{p}\E^{\beta \Delta E}},\\
p(1 \rightarrow 1)&=\frac{1}{d_{k-1}}\tr \mathcal{A}=\frac{\widetilde{q}-p\E^{-\beta \Delta E}}{\widetilde{p}-p\E^{-\beta \Delta E}}.
\end{split}
\ee

Thanks to the above formulas we can rewrite the term $\sqrt{p(0 \rightarrow 0)p(1 \rightarrow 1)}$ from the condition~\eqref{bound} using probabilities $p,q$:
\be
\sqrt{p(0 \rightarrow 0)p(1 \rightarrow 1)} = \frac{\sqrt{(q-\widetilde{p}\E^{\beta \Delta E})(p-\widetilde{q}\E^{\beta \Delta E})}}{\left| p-\widetilde{p}\E^{\beta \Delta E} \right|}.
\ee

\subsection{Bounds for coherences vs relaxation times: simplified scenario}

Here we examine the connection between our bounds for coherences in the qubit case (qubit $\rho_S$ with energy levels $E_0$ and $E_1$ and Hamiltonian $H_S$) and the relaxation time $T_2$ \cite{Gorini1978149, PhysRevA.66.062113, Alicki_book}. We are going to assume that diagonals elements decay exponentially, and then look at the decay of the off-diagonal as a function of $T_2$ (transverse relaxation) and t. We thus assume the following relations as they give exponential decay to the thermal state from any initial state
\be
p(1 \rightarrow 0) = p(1)(1 - \E^{-\frac{t}{T_2}}),\\
p(0 \rightarrow 1) = p(0)(1 - \E^{-\frac{t}{T_2}}),
\label{relaxtime}
\ee
where $p(0)$ and $p(1) = 1 - p(0)$ are the elements of the corresponding Gibbs state for $H_S$, $\rho_{S} = \begin{bmatrix}
p(0) & 0\\
0 & p(1)
\end{bmatrix} = \frac{1}{E^{-\beta E_0}+E^{-\beta E_1}}\begin{bmatrix}
E^{-\beta E_0} & 0\\
0 & E^{-\beta E_1}
\end{bmatrix}$.
Let us recall now that the damping factor for coherences, presented in Eq. \eqref{case5}, is equal to $\kappa = \sqrt{p(0 \rightarrow 0)p(1 \rightarrow 1)} = \frac{\sqrt{(q-\widetilde{p}\E^{\beta \Delta E})(p-\widetilde{q}\E^{\beta \Delta E})}}{\left| p-\widetilde{p}\E^{\beta \Delta E} \right|}$, $\widetilde{q} = 1 - q$, $\widetilde{p} = 1- p$, $\E^{\pm\beta \Delta E} = \E^{\pm\beta (E_j-E_i)}$ with $E_i$ being energy of the system and $\beta$ inverse temperature $\beta = \frac{1}{kT}$. We say that we need to express $\sqrt{p(0 \rightarrow 0)p(1 \rightarrow 1)}$ in terms of the relaxation time factor $\E^{-\frac{t}{T_2}}$. To do this, we insert relations from Eq. \eqref{relaxtime} into the constraints for transition probabilities coming from the unitary constrain, namely: 
\be
p(1 \rightarrow 1) + p(1 \rightarrow 0) = 1, \\
p(0 \rightarrow 0) + p(0 \rightarrow 1) = 1.
\label{uc}
\ee
We then have that the damping factor for coherences can be expressed as follows
\be
\label{decaytime}
\begin{split}
&\kappa = \sqrt{p(0 \rightarrow 0)p(1 \rightarrow 1)} = (p(0) - p(0)^2 + (1-2p(0)+2p(0)^2)E^{-\frac{t}{T_2}} + (p(0) - p(0)^2)E^{-2\frac{t}{T_2}})^{\frac{1}{2}} \\
&\leq \sqrt{p(0) - p(0)^2} + \sqrt{(1-2p(0)+2p(0)^2)}E^{-\frac{t}{2T_2}} + \sqrt{(p(0) - p(0)^2)}E^{-\frac{t}{T_2}}.
\end{split}
\ee
On the other hand, if instead we had decayed to the Gibbs state under the action of a Linblad generator (see eg \cite{Roga2010}), then we would have expected exponential decay of the off-diagonal terms
going as $\kappa=e^{-t/T_1}$, and $2T_1\geq T_2$. Here, we see that the decay of the off-diagonal terms does not decay exponentially fast, and does not even decay to zero as it would under a Lindblad generator. This is because we consider the class of optimal processes - the ones which preserve coherences as much as possible. 

Alternatively, we can derive this, using tools from open-system dynamics \cite{Alicki_book, Roga2010}. Since, as we have shown, Thermal Operations obeys the same block-diagonal structure (for non-degenerate Hamiltonians) as the Linblad generator under Davies maps,  we can write the corresponding one-qubit map in the computational basis for transition probabilities $p(i \rightarrow j)$ as:
\be
\begin{bmatrix}
p(0 \rightarrow 0) & 0 & 0 & p(0 \rightarrow 1)\\
0 & \sqrt{p(0 \rightarrow 0)p(1 \rightarrow 1)} & 0 & 0\\
0 & 0 & \sqrt{p(0 \rightarrow 0)p(1 \rightarrow 1)} & 0\\
p(1 \rightarrow 0) & 0 & 0 & p(1 \rightarrow 1)
\label{decay}
\end{bmatrix},
\ee
which we can compare with the standard one-qubt map from the Linblad generator (taken from \cite{Roga2010})
\be
\label{decay1}
\begin{bmatrix}
1-(1-E^{-\frac{t}{T_2}})p(1) & 0 & 0 & (1-E^{-\frac{t}{T_2}})p(0)\\
0 & E^{-\frac{t}{T_1}} & 0 & 0\\
0 & 0 & E^{-\frac{t}{T_1}} & 0\\
(1-E^{-\frac{t}{T_2}})p(1) & 0 & 0 & 1-(1-E^{-\frac{t}{T_2}})p(0)\\
\end{bmatrix}.
\ee
We can then assume an exponential decay of the diagonal terms, although not the off-diagonal elements (coherences). We can thus modify the above matrix from Eq. \eqref{decay1} to
\be
\begin{bmatrix}
1-(1-E^{-\frac{t}{T_2}})p(1) & 0 & 0 & (1-E^{-\frac{t}{T_2}})p(0)\\
0 & \sqrt{p(0 \rightarrow 0)p(1 \rightarrow 1)} & 0 & 0\\
0 & 0 & \sqrt{p(0 \rightarrow 0)p(1 \rightarrow 1)} & 0\\
(1-E^{-\frac{t}{T_2}})p(1) & 0 & 0 & 1-(1-E^{-\frac{t}{T_2}})p(0)\\
\label{decay2}
\end{bmatrix}.
\ee
Comparing the matrix from Eq. \eqref{decay} with that from Eq.\eqref{decay2}, we see that
\be
p(1 \rightarrow 0) = p(1)(1 - \E^{-\frac{t}{T_2}}),\\
p(0 \rightarrow 1) = p(0)(1 - \E^{-\frac{t}{T_2}}),
\ee 
And now, we can proceed as previously, namely as in Eqs. \eqref{uc} and \eqref{decaytime} to compute the damping factor $\kappa = \sqrt{p(0 \rightarrow 0)p(1 \rightarrow 1)}$ in terms of the relaxation time $T_2$. 

Let us now try to compare our results with that from open system dynamics. There, the coherences decay (exponentially) with respect to the time $T_1$, but we have found that this exponential decay is not necessary, since we obtain decay of coherences in terms of the time $T_2$. This suggest that there may be a way to engineer interactions which do not have $T_1$ times, but instead have a longer persisting coherence (recall that in our optimal processes, we have 
non-vanishing of the coherences). Also in a Markov process, one continuously evolves towards a Gibbs state (stationary state), which obviously have no coherences, where in our process, we can evolve towards a state that has the Gibbs' distribution on the diagonal (the same diagonal elements as the Gibbs state), but also have some non-zero off-diagonal elements - coherences. This difference partly comes from the fact that under Thermal Operations, 
we are applying a unitary which lasts some finite time, rather than some persistent interaction (as in a Markov process, and in principle, in open systems). Still, one could imagine repeating the Thermal Operations map, over and over, since the diagonals will continue to exponentially decay in each application of the map and coherences will be damped but not completely destroyed.

On the other hand, if we assume that both the diagonal of the density matrix, and its coherences decay exponentially,
then the dynamics is a semigroup $E^{Lt}$, where $L$ is the Lindblad operator, and we known expressions for relaxation times $T_1$ (longitudinal relaxation) and $T_2$ for semigroups and the relation between them: $2T_1 \geq T_2$. However, Thermal Operations do not form a semigroup, it is a wider class of operations, thus to have a full interpretation, going beyond our simplified scenario, one need to consider unitaries (that commute with the total Hamiltonian H) of the form $U = E^{-iHt}$ to have time explicitly in the formalism (and then, we can interpret what we mean by $T_1$ and $T_2$). This question goes beyond the scope of this paper, because our results concern the question of "whether" one can go from state 1 to state 2 (a question related to the second laws of thermodynamics), while the question of $T_1$ and $T_2$ concern themselves with the rate at which one goes from state 1 to state 2 (more related to questions in the spirit of the third law of thermodynamics (which concerns itself with how quickly you can go to the ground state).  We hope that this can open a new route of studies in engineering of thermodynamical processes. 

\subsection{No-go for higher dimensional systems?}
\label{beyond}
In this section we present results conjecturing that for higher dimension states $(d>2)$ it is impossible to saturate the bounds for processing of coherences. We do this by considering a class of Gibbs-preserving processes called quasi-cycles from which we choose one particular as our case-study example. The quasi-cycle can be defined as follows, we choose an order of levels and put them on a circle, fixing a direction. The aim of the process is to take all states from the group of states with the largest energy and shift them to the states with the energy level in the chosen direction.

Our case-study example is a three-level system and a quasi-cycle from Figure~\ref{fig:diag5}.
\begin{figure}[ht]
\begin{center}
\includegraphics[scale=0.25]{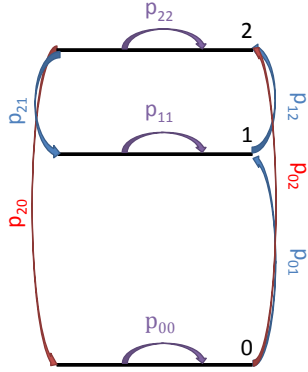}
\end{center}
\caption{\textbf{(Color online)} A $0 \rightarrow 1 \rightarrow 2$ three-level quasi-cycle (with energy level 0, 1 and 2) for which it is showed that the fundamental limit for coherences transport can not be reached. The quasi-cycle is set so forbidden transitions are $2\rightarrow 2$, $0 \rightarrow 2$, $2 \rightarrow 1$ and $1\rightarrow 0$, which means that their corresponding probabilities $p_{22}, p_{21}, p_{10}$ and $p_{02}$ are equal to zero, and the probability of transition $2 \rightarrow 0$, $p_{20}$, is equal to 1.}
\label{fig:diag5}
\end{figure}

At the beginning, let us show that there is a least one family of initial and final states (with given diagonal elements) for which the realization of the quasi-cycle from Fig.~\ref{fig:diag5} is unique.
\begin{fact}
\label{effect}
No other process has the same effect for a qutrit-to-qutrit transition between a family of states with given diagonal elements $(0,\frac{1}{2},\frac{1}{2}) \rightarrow (\frac{\E^{-\beta \Delta E_{21}}}{2}, \frac{1- \E^{-\beta \Delta E_{21}} + \E^{-\beta \Delta E_{20}}}{2}, \frac{1 - \E^{-\beta \Delta E_{20}}}{2})$ as the quasi-cycle from Fig. \ref{fig:diag5}.
\end{fact}
Before proving this, let us observe the following
\begin{lemma}
\label{fixed}
In the quasi-cycle from Fig.~\ref{fig:diag5} all probabilities are constrained and fixed. The are
\be
\begin{split}
\label{finalProb}
&p(2 \rightarrow 2)=0, p(2 \rightarrow 1) = 0, p(2 \rightarrow 0) = 1, \\
&p(1 \rightarrow 2) = \E^{-\beta \Delta E_{21}}, p(1 \rightarrow 1) = 1- \E^{-\beta \Delta E_{21}}, p(1 \rightarrow 0) = 0, \\
&p(0 \rightarrow 2)=0, p(0 \rightarrow 1)= \E^{-\beta \Delta E_{20}}, p(0 \rightarrow 0) = 1 - \E^{-\beta \Delta E_{20}}.
\end{split}
\ee
\end{lemma}
\begin{proof}
Let us start with writing conditions for probabilities of level transitions coming from the unitarity constrains
\be
\label{un1}
\begin{split}
&p(2 \rightarrow 2) + p(2 \rightarrow 1) + p(2 \rightarrow 0) = 1, \\
&p(1 \rightarrow 2) + p(1 \rightarrow 1) + p(1 \rightarrow 0) = 1, \\
&p(0 \rightarrow 2) + p(0 \rightarrow 1) + p(0 \rightarrow 0) = 1.
\end{split}
\ee
From the preservation of a Gibbs state, we also have that
\be
\label{un2}
\begin{split}
&p(2 \rightarrow 2) + p(1 \rightarrow 2)\E^{-\beta \Delta E_{12}} + p(0 \rightarrow 2)\E^{-\beta \Delta E_{02}} = 1, \\
&p(2 \rightarrow 1)\E^{-\beta \Delta E_{21}} + p(1 \rightarrow 1) + p(0 \rightarrow 1)\E^{-\beta \Delta E_{01}} = 1, \\
&p(2 \rightarrow 0)\E^{-\beta \Delta E_{20}} + p(1 \rightarrow 0)\E^{-\beta \Delta E_{10}} + p(0 \rightarrow 0) = 1,
\end{split}
\ee
where $\Delta E_{ij}$ is an energy difference between levels $i$ and $j$ of a qutrit.
Comparing Eq. \eqref{un1} with \eqref{un2} we obtain that
\be
\label{un3}
\begin{split}
&p(1 \rightarrow 2)\E^{-\beta \Delta E_{12}} + p(0 \rightarrow 2)\E^{-\beta \Delta E_{02}} = p(2 \rightarrow 1) + p(2 \rightarrow 0), \\
&p(2 \rightarrow 1)\E^{-\beta \Delta E_{21}} + p(0 \rightarrow 1)\E^{-\beta \Delta E_{01}} = p(1 \rightarrow 2) + p(1 \rightarrow 0), \\
&p(2 \rightarrow 0)\E^{-\beta \Delta E_{20}} + p(1 \rightarrow 0)\E^{-\beta \Delta E_{10}} = p(0 \rightarrow 2) + p(0 \rightarrow 1).
\end{split}
\ee
For our quasi-cycle, $p(2 \rightarrow 2), p(2 \rightarrow 1), p(1 \rightarrow 0)$, and $p(0 \rightarrow 2) = 0$, which immediately imposes $p(2 \rightarrow 0) = 1$. Then, inserting it into Eqs (\ref{un1}, \ref{un2}, \ref{un3}) and solving them, we get that all other probabilities are fixed too and given by
\be
\begin{split}
&p(2 \rightarrow 2)=0, p(2 \rightarrow 1) = 0, p(2 \rightarrow 0) = 1, \\
&p(1 \rightarrow 2) = \E^{-\beta \Delta E_{21}}, p(1 \rightarrow 1) = 1- \E^{-\beta \Delta E_{21}}, p(1 \rightarrow 0) = 0, \\
&p(0 \rightarrow 2)=0, p(0 \rightarrow 1)= \E^{-\beta \Delta E_{20}}, p(0 \rightarrow 0) = 1 - \E^{-\beta \Delta E_{20}}.
\end{split}
\ee
q.e.d. It implies that there is no freedom in choosing the rest of probabilities, the ones that are set to 0 and 1 already constrain and fix the rest.
\end{proof}
With the above, the proof of Fact \ref{effect} is quite straightforward. From the relation $\sum_i d_i p(i \rightarrow j) = d_j$ \cite{Horodecki_2013thermo}, we can build a stochastic matrix with probabilities $p(i \rightarrow j)$, which tells us whether, under a given input, the Gibbs state is preserved on the diagonal of a state. This effectively tells us which state-to-state transformations are possible under a given quasi-cycle from the point of view of their diagonal inputs (preservation of the Gibbs state). For our state we have,
$$\begin{bmatrix}
p(2 \rightarrow 2) & p(2 \rightarrow 1) & p(2 \rightarrow 0)\\
p(1 \rightarrow 2) & p(1 \rightarrow 1) & p(1 \rightarrow 0)\\
p(0 \rightarrow 2) & p(0 \rightarrow 1) & p(0 \rightarrow 0)\\
\end{bmatrix} \begin{bmatrix}
0\\
\frac{1}{2}\\
\frac{1}{2}\\
\end{bmatrix} = \begin{bmatrix}
\frac{\E^{-\beta \Delta E_{21}}}{2}\\
\frac{1- \E^{-\beta \Delta E_{21}} + \E^{-\beta \Delta E_{20}}}{2}\\
\frac{1 - \E^{-\beta \Delta E_{20}}}{2}
\end{bmatrix},$$ where, to obtain the final values, we put the probabilities from Eq. \eqref{finalProb}.

Due to the constrains on the matrix from the proof of Lemma \ref{fixed} (and the unitary matrix, constructed later in the txt, from $p(i \rightarrow j)$ from Eq. \eqref{uni1}), there is no freedom in changing $p(i \rightarrow j)$, which proves the uniqueness.

Before going further with the analyze of the state-to-state transitions, let us recall some auxiliary lemma proved by von Neumann \cite{vNeumann_1937matrix} and Fan \cite{Fan_1951matrix}, which will appear to be crucial in our further considerations. The lemma gives a maximization over $\tr XY^{\dagger}$, where $X,Y$ are some matrices, with respect to all possible rotations over $X$ and $Y$.
\begin{lemma}
\label{byNeumann}
If $X$ and $Y$ are $n \times n$ complex matrices, $W$ and $V$ are $n \times n$ unitary matrices, and $\sigma_1 \geq \cdots \geq \sigma_n \geq 0$ denotes ordered singular values, then
\be
\label{N1}
|\tr WXVY|\leq \sum_{i=1}^n\sigma_i(X)\sigma_i(Y)
\ee
and
\be
\label{N2}
\mathop{\operatorname{sup}}\limits_{W,V} |\tr WXVY|=\sum_{i=1}^n \sigma_i(X)\sigma_i(Y).
\ee
\end{lemma}

We are ready now to summarize our findings in the following
\begin{lemma}
\label{nogo}
Consider a unitary matrix $U = \bigoplus_k U_k$, written in the block form, where for a fixed block $k$, one has
\newcommand*{\tempbbq}{\multicolumn{1}{|c}{u^k_{(10)}}}
\newcommand*{\tempddq}{\multicolumn{1}{|c}{u^k_{(20)}}}
\newcommand*{\tempbbb}{\multicolumn{1}{|c}{u^k_{(11)}}}
\newcommand*{\tempddd}{\multicolumn{1}{|c}{u^k_{(21)}}}
\newcommand*{\tempoo}{\multicolumn{1}{|c}{u^k_{(01)}}}
\newcommand*{\tempxx}{\multicolumn{1}{|c}{u^k_{(00)}}}
\be
U_k=\begin{pmatrix} u^k_{(22)} & \tempddd & \tempddq \\ \hline u^k_{(12)} & \tempbbb & \tempbbq \\ \hline u^k_{(02)} &\tempoo & \tempxx \end{pmatrix},
\ee where, for each $k$, $u^k_{(ij)}$ is a matrix of dimension $d_i \times d_j$. Assuming that the dimensions are such that $d^k_0 > d^k_1 > d^k_2$, and $u^k_{(22)} = 0, u^k_{(21)} = 0, u^k_{(10)} = 0, u^k_{(02)} = 0$, and $\tr u_{00} u^{\dagger}_{00} \neq \tr u_{11} u^{\dagger}_{11}$, there is no such a $U$ that saturates the Cauchy-Schwarz inequality of the form $\tr u_{(00)}^k {u_{(11)}^{k+l}}^{\dagger} \leq \sqrt{\tr u_{(00)}^k {u_{(00)}^{k}}^{\dagger} \tr u_{(11)}^{k+l} {u_{(11)}^{k+l}}^{\dagger}}$, where $l$ is an integer, such that $d^k_0 = d^{k+l}_1$, one always has a strict inequality.
\end{lemma}

Lemma~\ref{nogo} implies the following:
\begin{corollary}
\label{nogo1}
There in no such an energy-preserving unitary $U$ that commutes with the total Hamiltonian of the system-bath setup, where one has a generic heat bath that follows assumptions from Secs \ref{App:not} and \ref{App:add} from Appendix and realizes the state-to-state from transition from Fact \ref{effect} in a precise way (no disturbance and approximations in reaching the final state) that leads to the best possible processing of coherences, which means saturation of the bounds from Proposition \ref{p1} (in that case $|\alpha|=\sqrt{p(0 \rightarrow 0)p(1 \rightarrow 1)}$).
\end{corollary}

To prove Corollary~\ref{nogo1}, one needs to adapt the mathematical structures from Lemma~\ref{nogo} to states transitions under Thermal Operations and conect it with the facts already shown in this section.
\begin{enumerate}
	\item The unitary matrix from Lemma~\ref{nogo} can be treated as an energy-preserving matrix that is used to implement state-to-state transitions under Thermal Operations as in Eq. \eqref{CPTP}.
	\item The channel that realizes state-to-state transition from Fact \ref{effect} is unique, and its transition probabilities $p(i \rightarrow j)$ correspond to terms $\tr u_{ij} u^{\dagger}_{ij}$ from Lemma~\ref{nogo}.
	\item The generic heat bath (laws of degenerations from Secs \ref{App:not} and \ref{App:add}) gives that the dimensions of blocks ($d^k_0 > d^k_1 > d^k_2$) should be strict ineqaulities.
	\item The Cauchy-Schwarz inequality can be idenfity with the bounds from Proposition \ref{p1} (since, in this state-to-state transition, $p(2 \rightarrow 2)$, there is only one bound that one can try to saturate - $|\alpha|=\sqrt{p(0 \rightarrow 0)p(1 \rightarrow 1)}$).
\end{enumerate}

\begin{proof}[{\bf Proof of Lemma~\ref{nogo}}]
Let us fix $k^{\text{th}}$ energy block, then a general unitary transformation has a form
\newcommand*{\tempbbq}{\multicolumn{1}{|c}{u^k_{(10)}}}
\newcommand*{\tempddq}{\multicolumn{1}{|c}{u^k_{(20)}}}
\newcommand*{\tempbbb}{\multicolumn{1}{|c}{u^k_{(11)}}}
\newcommand*{\tempddd}{\multicolumn{1}{|c}{u^k_{(21)}}}
\newcommand*{\tempoo}{\multicolumn{1}{|c}{u^k_{(01)}}}
\newcommand*{\tempxx}{\multicolumn{1}{|c}{u^k_{(00)}}}
\be
\label{uni1}
U_k=\begin{pmatrix} u^k_{(22)} & \tempddd & \tempddq \\ \hline u^k_{(12)} & \tempbbb & \tempbbq \\ \hline u^k_{(02)} &\tempoo & \tempxx \end{pmatrix},
\ee
where numbers in the brackets correspond to transitions between levels of our system. To obtain the result, we need a simpler form of the matrix $U_k$. Thanks to Lemma~\ref{byNeumann} we know that the maximal value of $\tr u_{(00)}^ku_{(11)}^{k+l}$ is equal to $\sum_i \sigma_i\left(u_{(00)}^k\right)\sigma_i\left(u_{(11)}^{k+l}\right)$, where singular values are taken is a non-increasing order. This allows us to consider only singular values of $u_{(00)}^k$ and $u_{(11)}^{k+l}$, because we want to know maximal possible values of trace and compare it with the bound that comes from Proposition \ref{p1}. From the general theory we know that it saturates when either the first or the second vector is a multiple of the other. So, to obtain the result, we have to show that the vector constructed from non-increasing ordered singular values of $u_{(00)}^k$ is not proportional to the vector constructed in the same way from the block $u_{(11)}^{k+1}$. We show this using an explicit form of our quasicycle from Figure~\ref{fig:diag5} and unitary constraints $U_kU_k^{\dagger}=U_k^{\dagger}U_k=\text{\noindent
\(\mathds{1}\)}$. From the form of our quasicycle one can see that blocks which correspond to zero probabilities of transition are represented by zero matrices. Indeed constraints $\tr (u^k_{(ij)}(u^{k}_{(ij)})^{\dagger})=0$ implies that $u^k_{(ij)}=\mathcal{O}$, where $\mathcal{O}$ denotes zero matrix. In our case we have that $u^k_{(22)}=u^k_{(21)}=u^k_{(10)}=u^k_{(20)}=\mathcal{O}$ and the matrix $U_k$ form~\eqref{uni1} looks like
\newcommand*{\tempbbqQ}{\multicolumn{1}{|c}{\mathcal{O}}}
\newcommand*{\tempdddQ}{\multicolumn{1}{|c}{\mathcal{O}}}
\be
\label{uni2}
U_k=\begin{pmatrix} \mathcal{O} & \tempdddQ & \tempddq \\ \hline u^k_{(12)} & \tempbbb & \tempbbqQ \\ \hline \mathcal{O} &\tempoo & \tempxx \end{pmatrix}.
\ee
In the next step we use singular value decomposition (SVD) to $u^k_{(11)}$ and $u^k_{(00)}$. Thanks to this we can write $u_{(11)}^k=A^k_{(11)}\Sigma^k_{(11)}(B^k_{(11)})^{\dagger}$ and $u_{(00)}^k=A^k_{(00)}\Sigma^k_{(00)}(B^k_{(00)})^{\dagger}$, where $A^k_{(ii)},B^k_{(ii)}$ for $i=0,1$ are rectangular, unitary matrices and $\Sigma^k_{(00)}, \Sigma^k_{(11)}$ are square, diagonal matrices with singular values as entries. Using SVD we can define new unitary matrix $\widetilde{U}_k$ which gives us the same probability transitions (since $\tr u_{ij} u^{\dagger}_{ij} = \tr \Sigma_{ij}\Sigma_{ij}$), but it is simpler to analysis:
\newcommand*{\tempdddQQa}{\multicolumn{1}{|c}{(A^k_{(11)})^{\dagger}}}
\newcommand*{\tempdddQQaa}{\multicolumn{1}{|c}{(A^k_{(00)})^{\dagger}}}
\newcommand*{\tempdddQQb}{\multicolumn{1}{|c}{B^k_{(11)}}}
\newcommand*{\tempdddQQbb}{\multicolumn{1}{|c}{B^k_{(00)}}}
\newcommand*{\tempdddQQx}{\multicolumn{1}{|c}{(A^k_{(22)})^{\dagger}u^k_{(20)}B^k_{(00)}}}
\newcommand*{\tempdddQQxx}{\multicolumn{1}{|c}{(A^k_{(00)})^{\dagger}u^k_{(01)}B^k_{(11)}}}
\newcommand*{\tempdddQQs}{\multicolumn{1}{|c}{\Sigma^k_{(00)}}}
\newcommand*{\tempdddQQss}{\multicolumn{1}{|c}{\Sigma^k_{(11)}}}
\be
\label{tildeU}
\begin{split}
\widetilde{U}_k&=\begin{pmatrix} (A^k_{(22)})^{\dagger} & \tempdddQ & \tempdddQ \\ \hline \mathcal{O} & \tempdddQQa & \tempbbqQ \\ \hline \mathcal{O} &\tempdddQ & \tempdddQQaa \end{pmatrix} \begin{pmatrix} \mathcal{O} & \tempdddQ & \tempddq \\ \hline u^k_{(12)} & \tempbbb & \tempbbqQ \\ \hline \mathcal{O} &\tempoo & \tempxx \end{pmatrix} \begin{pmatrix} B^k_{(22)} & \tempdddQ & \tempdddQ \\ \hline \mathcal{O} & \tempdddQQb & \tempbbqQ \\ \hline \mathcal{O} &\tempdddQ & \tempdddQQbb \end{pmatrix}\\
&=\begin{pmatrix} \mathcal{O} & \tempdddQ & \tempdddQQx \\ \hline (A^k_{(11)})^{\dagger}u^k_{(12)}B^k_{(22)} & \tempdddQQss & \tempbbqQ \\ \hline \mathcal{O} &\tempdddQQxx & \tempdddQQs \end{pmatrix}.
\end{split}
\ee
Indeed such transformation  ensures that main blocks of $\widetilde{U}_k$ have diagonal form, namely all diagonal blocks are equal to $\Sigma^k_{(ii)}$.

Now we show that singular values of the main blocks are equal to zero or one. Because now we deal with nonzero probabilities it is obvious that some of the singular values of the main have to be strictly positive. It is also important to mention that without lose of generality we rearrange rows of $\widetilde{U}_k$ in such a way that singular values are in decreasing order. We also interpret rows and columns of matrix $\widetilde{U}_k$ as a vectors $|r_s\>$ and $|c_l\>$ respectively.
Because of unitarity conditions these vectors have to be orthonormal, i.e. $\<r_s|r_l\>=\delta_{sl}$ and $\<c_s|c_l\>=\delta_{sl}$. Let us take now row $|r_l\>$ which posses nonzero singular value $\sigma_l$, for example from main block $\Sigma^k_{(11)}$. Computing scalar product of this vector $|r_l\>$ with any other row vector $|r_j\>$, where $d_k+d_{k-1} < j \leq d_k+d_{k-1}+d_{k-2} $, together with above mentioned condition $\sigma_l \neq 0$ we can conclude that column vector which contains singular value $\sigma_l$ has only one nonzero element which is our $\sigma_l$. Using normalization constraint $\<r_j|r_j\>=1$ we have that $\sigma_l=1$. This same argumentation  can be used to the rest of nonzero singular values and see that only possible values of all singular values are zero or one.

In the last step  we have to show that vector constructed from non-increasing ordered singular values of $u_{(00)}^k$ is not proportional to vector constructed in the same way, but from the block $u_{(11)}^{k+l}$. We know that these two blocks determine different probabilities $p_{00}$ and $p_{11}$ which due to equation~\eqref{finalProb} have to be different. Together with knowledge that all positive singular value are equal to one we can say that vectors constructed in an aforementioned way have different length so they cannot saturate the bound form Eq.~\ref{p1}.
\end{proof}

\begin{remark}
One can notice that one of the possible realization of unitary transformation $U_k$ in Eq.~\eqref{uni1} for the quasi-cycle from Fig.~\ref{fig:diag5} is a realization in the so-called "brute" form. It means that all nonzero elements of $U_k$ are equal to one. Then every such a transformation can be written as a direct sum of permutation matrices
\be
U_k=\left( \bigoplus_l V(\pi)\right) \oplus \text{\noindent
\(\mathds{1}\)}
\ee
for a certain $\pi \in S(3)$. The identity follows from the fact that some of the levels are untouched. Of course, that realization is quite harmful for coherences, and is far from saturating the inequality.
\end{remark}

In the end, we want to state the following conjecture 
\begin{conjecture}
Corrolary \ref{nogo1} is true also if the transition $\rho \rightarrow \sigma$ is realized in the perturbed way, and instead of the final state $\sigma$ one obtains a state $\sigma'$, such that $|\sigma - \sigma'| < \delta$, where $\delta$ is small, i.e. probabilities of transitions $p(i \rightarrow j)$ that previously were equal to 0, now are equal to $p(i \rightarrow j) = \epsilon$, and other probabilities are also respectively change. Summing up, there is no such a channel that realizes the state-to-state transition from Fact \ref{effect} in the most friendly way for coherences, one is not able to reach the fundamental limit for a limimal coherences damping.
\end{conjecture}
The probabilities look then as follows
\begin{fact}
When we set the probabilities $p(2 \rightarrow 2), p(2 \rightarrow 1), p(1 \rightarrow 0) and p(0 \rightarrow 2)$ (that previously, in the exact state-to-state transition were equal to 0) to be all equal to $\epsilon$, where $\epsilon$ is small, the other probabilities of the perturbed version of the quasi-cycle $2\rightarrow1\rightarrow0$ from Fig. \ref{fig:diag5} are all fixed and equal to
\be
\begin{split}
&p(2 \rightarrow 2)=\epsilon, p(2 \rightarrow 1) = \epsilon, p(2 \rightarrow 0) = 1 - 2\epsilon, \\
&p(1 \rightarrow 2) = \E^{-\beta \Delta E_{21}}(1-\epsilon)-\E^{-\beta \Delta E_{10}}\epsilon, p(1 \rightarrow 1) = (1- \E^{-\beta \Delta E_{21}})(1-\epsilon)-\E^{\beta \Delta E_{10}}\epsilon, p(1 \rightarrow 0) = \epsilon, \\
&p(0 \rightarrow 2)=\epsilon, p(0 \rightarrow 1)= (\E^{-\beta \Delta E_{20}})(1-2\epsilon)+(\E^{-\beta \Delta E_{10}}-1)\epsilon, p(0 \rightarrow 0) = 1 - \epsilon - (1 - \epsilon)\E^{-\beta \Delta E_{20}}+\E^{\beta \Delta E_{10}}\epsilon.
\end{split}
\ee
\end{fact}
This construction is sufficient to search for a counterexample (that in this quasi-cycle, the previously impossible saturation of the bound for coherences is possible).
We examine many construction, which should be the most crude forms of a perturbation of the unitary matrix, ie, in blocks $u^k_{ij}$ of U that correspond to probabilities equal to zero, we put some number of perturbation represented by $\sqrt{\epsilon}$ on a diagonal of the block, and this construtions always lead us to the previously considered case (the one without the perturbation), since the matrix elements of the unitary matrix U are then equal to 0 or 1, or are $\epsilon$-closed. Of course, a deeper analysis is needed to find a counter-example, or to analytically verify our conjecture.

\bibliographystyle{apsrev}
\bibliography{rmp15-hugekey,../../common/refthermo}
\end{document}